\documentclass[12pt,reqno,a4paper]{amsart}
\usepackage{graphics}
\usepackage{amssymb}
\usepackage{enumerate}

\usepackage{amsmath}

\usepackage{color}
\usepackage{colordvi}

\numberwithin{equation}{section}

 \theoremstyle{plain}            
 \newtheorem{theorem}{Theorem}[section]
 \newtheorem{proposition}[theorem]{Proposition}
 \newtheorem{lemma}[theorem]{Lemma}

 \theoremstyle{definition}       

 \newtheorem{remark}[theorem]{Remark}
 
\begin{document}
\title[A Product Formula Related to Quantum Zeno Dynamics]
{Note on a Product Formula Related \\[.2em] to Quantum Zeno Dynamics}
\author[P.\ Exner]{Pavel Exner}
\address{Doppler Institute for Mathematical Physics and Applied Mathematics\\
Czech Technical University\\
B\v rehov\'a 7\\ 11519 Prague\\ Czechia\\ and Department of
Theoretical Physics\\ NPI\\ Academy of Sciences\\ 25068 \v{R}e\v{z}
near Prague, Czechia} \email{exner@ujf.cas.cz}
\urladdr{http://gemma.ujf.cas.cz/~exner/}

\author[T.\ Ichinose]{Takashi Ichinose}
\address{Department of Mathematics \\ Faculty of Science \\ Kanazawa University \\ Kanazawa 920-1192 \\ Japan}
\email{ichinose@staff.kanazawa-u.ac.jp}

\maketitle

\begin{center}
\emph{To the memory of our friend Hagen Neidhardt (1950--2019)}
\end{center}

\begin{abstract}
Given a nonnegative self-adjoint operator $H$ acting on a separable Hilbert space and an orthogonal projection $P$ such that $H_P := (H^{1/2}P)^*(H^{1/2}P)$ is densely defined, we prove that $\lim_{n\rightarrow \infty} (P\,\mathrm{e}^{-itH/n}P)^n = \mathrm{e}^{-itH_P}P$ holds in the strong operator topology. We also derive  modifications of this product formula and its extension to the situation when $P$ is replaced by a strongly continuous projection-valued function satisfying $P(0)=P$.

\end{abstract}


\section{Introduction and the main result}
\label{s: intro}

The main aim of this paper is to prove the following product
formul{\ae} for short-time Schr\"odinger unitary groups and
orthogonal projections:

\begin{theorem} \label{th:main}
Let $H$
be a nonnegative self-adjoint operator acting on a separable Hilbert space $\mathcal{H}$ and $P$ an orthogonal projection onto a closed subspace of $\mathcal{H}$. Suppose that $H^{1/2}P$ is densely defined, so that $H_P := (H^{1/2}P)^*(H^{1/2}P)$ is a self-adjoint operator. Then for any $f\in \mathcal{H}$ and $\varepsilon = \pm 1$ the following relations hold,
\begin{align}
& \lim_{n\rightarrow \infty} (P\,\mathrm{e}^{-\varepsilon itH/n}
P)^nf
   = \mathrm{e}^{-\varepsilon itH_P}Pf\,, \label{Pboth} \\
& \lim_{n\rightarrow \infty} (\mathrm{e}^{-\varepsilon itH/n} P)^nf
   = \mathrm{e}^{-\varepsilon itH_P}Pf\,, \label{Pright} \\
& \lim_{n\rightarrow \infty} (P\,\mathrm{e}^{-\varepsilon itH/n})^nf
   = \mathrm{e}^{-\varepsilon itH_P}Pf \label{Pleft}
\end{align}
in the Hilbert space norm, and moreover, the convergence is uniform on every
bounded $t$-interval in $\mathbb{R}$.
\end{theorem}
\noindent Needless to say, the claim is nontrivial only if $H$ and $P$ do not
commute.

The main motivation to study such product formul{\ae} comes from the
behavior of quantum systems exposed to frequent measurements. Turing
was the first to notice \cite{Ho04} that if we ascertain repeatedly
whether a quantum system is in a given state, then in the limit
of infinite measurement frequency it becomes impossible to leave
this state. The idea was rediscovered in the context of unstable
system decays -- see, e.g., \cite{BN67, Fr72} -- but it attracted a
wide attention only after Misra and Sudarshan \cite{MS77} invented a
catchy name calling such a behavior \emph{quantum Zeno effect} in an
allusion to the classical Zeno aporia about a flying arrow. More
about the early history can be found in \cite{Ex85}.

The second breaking moment came in 1990 when Itano et al.
\cite{IHBW90} demonstrated the existence of the effect
experimentally. Since then it became object of an extensive
examination, both from the theoretical and experimental points of
view, and it led even to various practical applications; a partial
summary can be found in the review paper \cite{FP08}.

If the projection $P$ describing the measurement has a dimension
larger than one, the question about the time evolution in the
subspace to which the permanent observation confines the state of
the system becomes nontrivial. It is natural to expect that such a
`Zeno dynamics' will be governed by the part of the original
Hamiltonian $H$ acting in the subspace $P\mathcal{H}$, and it was
shown in \cite[Sect.~2.4]{Ex85} that the generator is indeed
associated with appropriate quadratic form constructed from the
operator $H$ and the projection $P$. It was not easy, however, to
establish the existence of the Zeno dynamics beyond the situations
when the dimension of $P$ is finite or the operator $H$ is bounded;
needless to say that this is often not the case with actual physical
systems. A prime example is a permanent ascertaining whether a free
quantum particle dwells within a prescribed region $\Omega$ of the
configuration space discussed in \cite{FPSS01}, see also
\cite{FP08}, with the conclusion that the Zeno generator is (the
multiple of) the corresponding Dirichlet Laplacian. The argument
made use of the stationary phase method but the existence of the
limit was not actually established by the authors.

Motivated by the said paper we addressed the question of the Zeno
dynamics existence in \cite{EI05}, where we have managed to
establish the existence of the limits of the expressions appearing
on the left-hand sides of \eqref{Pboth}--\eqref{Pleft} in the
topology of a larger space, namely, the Fr\'echet space
$L^2_{\text{\rm loc}}({\mathbb R}; {\mathcal H}) = L^2_{\text{\rm
loc}}({\mathbb R})\otimes {\mathcal H}$, provided that $H$ is
semibounded and the operator $H_P$ is densely defined; the validity
of the formul{\ae} is preserved if the exponential in
\eqref{Pboth}--\eqref{Pleft} is replaced by functions of a wider
class, in particular, by the resolvent $(I+itH)^{-1}\,$
\cite{EINZ07, Ich15}. We argued in \cite{EI05} that such a result
can be regarded as sufficient from the viewpoint of physics due to
the fact that every measurement, in particular, that of time is
burdened with errors, and any actual experiment typically involves
averaging over a large number of system copies.

It is desirable, though, to answer the question without such a
underpinning by demonstrating the result with the convergence in a
stronger sense, namely that of the strong operator topology. This is
the aim of the present paper. In addition to the described physical
motivation, the obtained relation are of independent mathematical
interest belonging to the genre of the product formul{\ae} of
Trotter and Trotter-Kato \cite{Ka78, Tr59}, see also \cite{RS80}. In
fact, we are going to prove a slightly more general claim with a fixed
$P$ replaced by a projection-valued function of~$t$ satisfying
certain regularity assumptions.

\begin{theorem} \label{th:main'}
Let $\mathcal{H}$ be a separable Hilbert space, and $H,\,P$, and thus also $H_P$, be the same as in Theorem~\ref{th:main}. Let further $P(\cdot)$ be a strongly continuous function the values of which are orthogonal projections in $\mathcal{H}$ defined in a right neighbourhood of zero and satisfying $P(0)=P$. Moreover, suppose that
\begin{equation} \label{hypoth-to-H1/2P}
\lim_{\tau\to 0+} [\tau^{-1}(I-\mathrm{e}^{-it\tau H})]^{1/2}P(\tau) v
= \mathrm{e}^{\pi i/4}(tH)^{1/2}Pv\,,
\end{equation}
for every $v \in D[H^{1/2}P]$. Then for any $f \in \mathcal{H}$ and $\varepsilon = \pm 1$ we have
\begin{align}
& \lim_{n\rightarrow \infty} (P(1/n)\,\mathrm{e}^{-\varepsilon
itH/n} P(1/n))^nf
   = \mathrm{e}^{-\varepsilon itH_P}Pf\,, \label{P(t)both} \\
& \lim_{n\rightarrow \infty} (\mathrm{e}^{-\varepsilon itH/n}
P(1/n))^nf
   = \mathrm{e}^{-\varepsilon itH_P}Pf\,, \label{P(t)right} \\
& \lim_{n\rightarrow \infty} (P(1/n)\,\mathrm{e}^{-\varepsilon
itH/n})^nf
   = \mathrm{e}^{-\varepsilon itH_P}Pf\,, \label{P(t)left}
\end{align}
in the Hilbert space norm, where the convergence is uniform on every bounded $t$-interval in $\mathbb{R}$.
\end{theorem}

\begin{remark} \label{}
Note that the hypothesis made in Theorem~\ref{th:main'} about the convergence of $[\tau^{-1}(I- \mathrm{e}^{-\tau H})]^{1/2}P(\tau)v$ is slightly weaker in comparison with \cite[Theorem~2.1]{EI05} where we assumed that $D[H^{1/2}P(\tau)] \supset D[H^{1/2}P]$ and $\lim_{\tau \rightarrow 0+} \|H^{1/2}P(\tau)v\| = \|H^{1/2}Pv\|$ holds for every $v \in D[H^{1/2}P]$. This, in fact, was not fully necessary there as a footnote in \cite[p.~206]{EI05} briefly mentioned.
\end{remark}

Note that the assumption of positivity of $H$ is made for
convenience only, it is obvious that the result remains to be valid
if $H$ is replaced by $H+cI$ with a fixed $c\in\mathbb{R}$, i.e. for
any self-adjoint operator bounded from below. On the other hand, the
density hypothesis is crucial; in \cite[Rem.~2.7]{EI05} we cited an
example showing that in its absence the expressions
$(\mathrm{e}^{-itH/n} P)^n$ may not converge in any sense. It may
also happen that they converge but not strongly. Examples were found
by Matolcsi and Svidkoy \cite{MS03}, however, they do not contradict
Theorem~\ref{th:main'} because in one of them the analogue of $H_P$
is not densely defined and in the other the operator $H$ is not
semibounded.

Let us now describe briefly our strategy to prove Theorem~\ref{th:main'}. The main tool is Chernoff's theorem \cite{Ch74}, see also \cite{Ch68}, which for reader's convenience we reproduce in Sect.~\ref{s: proof} below. It will yield the sought result if we show that the $\tau$-family $\{[I+ \tau^{-1}(I-P(\tau)\,\mathrm{e}^{-\varepsilon it\tau H} P(\tau))]^{-1}\}_{\tau>0}$ converges to $(I+\varepsilon itH_P)^{-1}$ as $\tau\to 0+$ in the strong operator topology. The proof might basically follow the argument used by Kato in \cite{Ka78} to establish his celebrated {\it self-adjoint} Trotter-Kato product formula for the form sum of two nonnegative self-adjoint operators, since these two problems appear to have notable similarities. However, a straightforward analogy of Kato's argument is not sufficient due to a difficulty one encounters, to be specified in Sect.~\ref{s: proof}, Step IV of the proof. The point is that the argument \emph{can} be applied to a certain class of admissible functions $\phi(x)$ which contains beside real exponentials, e.g., $(1+\varepsilon ix)^{-1}$, as shown in the mentioned papers \cite{EINZ07, Ich15}, but unfortunately \emph{this class fails to include} $\mathrm{e}^{-\varepsilon ix}$ corresponding to our unitary group $\mathrm{e}^{-itH}$.

Our way to overcome this obstacle is to start from the weaker result mentioned above. What we did in \cite{EI05} was to complement the modified Kato's argument by the Vitali theorem from complex function theory; in that way we proved that for each $f \in {\mathcal H}$ the $\tau$-family $\{[I+ \tau^{-1}(I-P(\tau)\mathrm{e}^{-\varepsilon it\tau H}P(\tau))]^{-1}f\}_{\tau>0}$ converges in the the Fr\'echet space $L^2_{\text{\rm loc}}({\mathbb R}; {\mathcal H})$. This conclusion serves as a departing point here, although it only implies that for some set $M_f \subset [0,\infty)$ of Lebesgue measure zero, the $\tau$-family $\{[I+ \tau^{-1}(I-P(\tau)\mathrm{e}^{-\varepsilon it\tau H}P(\tau))]^{-1}f\}_{\tau>0}$ converges in the Hilbert space norm for every $t \in [0,\infty) \setminus M_f$, not excluding the possibility that the convergence \emph{does} hold at some points of $M_f$. Furthermore, using the separability hypothesis made about the Hilbert space ${\mathcal H}$, we can choose a countable dense subset ${\mathcal D} := \{f_l\}_{l=1}^{\infty}$ in ${\mathcal H}$. Putting then $M = M_{\mathcal D} := \cup_{l=1}^{\infty} M_{f_l}$, which is also a set of Lebesgue measure zero, we may say that the above indicated $\tau$-family converges for all $t \in {\mathbb R}\setminus M$ and for every $f \in {\mathcal D}$, and therefore, in view of the density, also for every $f \in {\mathcal H}$.

To pass from the `almost all $t$' to the `all $t$' stage, one has to demonstrate that the exceptional set $M$ is in fact empty. This task may seem a small step, but in reality it proved to be a deep and highly nontrivial question. Our way to deal with it is to establish the equicontinuity of the above $\tau$-family, cf.~Lemma~\ref{l:equiconti}, which would allow us to achieve our goal by means of the Ascoli-Arzel\`{a} theorem.

A detailed description of the argument we have sketched here, given
in Sect.~\ref{s: proof} and Sect.~\ref{s: proof-lemma}, is a core part
of the paper. As a preliminary, we characterize in the next section the limit
self-adjoint operator $H_P$ appearing in the theorems. Finally, the
paper will be concluded with a short section in which we return
briefly to example of the permanent position measurement considered
in \cite{FPSS01}.

\section{Concerning the limit self-adjoint operator~$H_P$}
\label{s: limitop}

Throughout the paper $H$ will be a nonnegative self-adjoint operator in a
separable Hilbert space ${\mathcal H}$, and $P$ will be an orthogonal
projection. As we have indicated above, the nonnegativity assumption
is made for convenience; our main result extends easily to any
self-adjoint operator $H$ bounded from below as well as, by sign
change, to one bounded from above, i.e. to each semi-bounded
self-adjoint operator in ${\mathcal H}$.

Our aim here is to elucidate what is the Zeno generator $H_P$ appearing in Theorems~\ref{th:main} and \ref{th:main'} and to indicate some of its properties which we will need in the sequel. Consider the quadratic form $u \mapsto \|H^{1/2}Pu\|^2$ with \emph{form domain} $D[H^{1/2}P]$, being the domain of the \emph{operator} $H^{1/2}P$. While the domains of $HP$ and $H^{1/2}P$ are nontrivial only as subspaces of $P{\mathcal H}$, we consider these operators always as acting in the whole Hilbert space ${\mathcal H}$ writing, if needed, $HP = HP\!\!\restriction\!_{P{\mathcal H}}\oplus\, 0\!\restriction\!_{(I-P){\mathcal H}}$ with the domain
\begin{equation}  \label{HP}
  D[HP] = D[HP\!\!\restriction\!_{P{\mathcal H}}]\oplus (I-P){\mathcal H}, \quad D[HP\!\!\restriction\!_{P{\mathcal H}}] = D[H]\cap P{\mathcal H},
\end{equation}
and similarly for $H^{1/2}P$. Note that the latter is a \emph{closed} linear operator, because $H^{1/2}$ is closed; from the same reason the operator $HP$ is also closed.  We have $(H^{1/2}P)^* \supset PH^{1/2}$ in general. On the other hand, the operator $PH^{1/2}$ may not be closed. This would be the case if and only if there were $c_0,\,c_1 > 0$
such that $\|H^{1/2}u\| \leq c_0\|PH^{1/2}u\| + c_1\|u\|$ holds for all $u$ in the domain $D[H^{1/2}]$ of $PH^{1/2}$, however, the necessary part of this condition may not be satisfied. The same is true for $PH$.

By $H_P$ we denote the unique self-adjoint operator associated with the above mentioned quadratic form. Its (operator) domain $D[H_P]$ is a subspace of $D[H^{1/2}P]$ consisting of all $u \in D[H^{1/2}P]$ which satisfy $|\langle H^{1/2}Pu, H^{1/2}Pv \rangle| \leq C\|v\|$ for all $v \in D[H^{1/2}P]$ with a constant $C \geq 0$, cf.~\cite[Sect.~VI.2.1]{Ka76}, so that
\begin{equation}  \label{H_P}
  H_P = (H^{1/2}P)^*(H^{1/2}P).
\end{equation}
Needless to say, the form domain of $H_P$ is the (operator) domain of $H_P^{1/2}$, for which we have by polar decomposition, cf.~\cite[Sect.~VI.2.7]{Ka76},
\begin{align}
&H_P^{1/2} = |H^{1/2}P| = [(H^{1/2}P)^*(H^{1/2}P)]^{1/2}\,,
                                                     \nonumber\\
&D[H_P^{1/2}] = D[H^{1/2}P]
= (D[H^{1/2}] \cap P{\mathcal H}) \oplus (I-P){\mathcal H}.
                                             \label{domain-H1/2P}
\end{align}
Then it is not difficult to check the following claim:

\begin{proposition} \label{p:H_P}
The operator $H_P$ in \eqref{H_P} decomposes as
\begin{equation} \label{decomp-HP}
H_P = H_P\!\!\restriction\!_{P{\mathcal H}}
                     \oplus\, 0\!\restriction\!_{(I-P){\mathcal H}}
\end{equation}
having the dense domain
 $$
D[H_P] := D[H_P\!\!\restriction\!_{P{\mathcal H}}] \oplus (I-P){\mathcal H}
 \,\subset\, P{\mathcal H} \oplus (I-P){\mathcal H} = {\mathcal H}\,.
$$
Here $H_P\!\!\restriction\!_{P{\mathcal H}}$ is the $P{\mathcal H}$-component of $H_P$, which is the self-adjoint operator in the subspace $P{\mathcal H}\subset\mathcal{H}$ associated with the quadratic form
in $P{\mathcal H}$ with the form domain $D[{H_P}^{1/2}]\cap P{\mathcal H}$,
\begin{equation} \label{quad-H1/2P}
D[{H_P}^{1/2}]\cap P{\mathcal H} = D[H^{1/2}P]\cap P{\mathcal H} \ni w
             \mapsto  \|H^{1/2}Pw\|^2 = \|H_P^{1/2}w\|^2\,,
\end{equation}
and the zero operator $0\!\!\restriction\!\!_{(I-P){\mathcal H}}$ is its $(I\!-\!P){\mathcal H}$-component, trivially bounded and self-adjoint on the subspace $(I\!-\!P){\mathcal H}$ orthogonal to $P\mathcal{H}$.
\end{proposition}

From Proposition~\ref{p:H_P}, we can see that $H_P$ is in general not a restriction of $H$, and furthermore, the inclusion $D[H_P] \subset D[H]$ does not hold either. This is obvious, since $H$ is the unique self-adjoint operator in ${\mathcal H}$ associated with the quadratic form
\begin{equation} \label{quad-H1/2}
 D[H^{1/2}] \ni w \mapsto  \|H^{1/2}w\|_{\mathcal H}^2\,,
\end{equation}
so that $(I-P){\mathcal H}$ is not a subset of $D[H]$ as long as $H$ is \emph{unbounded}.

\smallskip

We can make, however, a weaker claim described in the following proposition; we note in passing that it is well illustrated by the inclusion
$D[(-\Delta)_{\Omega}] \subset D[-\Delta]$ from the example treated in Sect.~\ref{s: example} below.

\begin{proposition} \label{p:DHP<DH}
Let $H$ be a nonnegative unbounded self-adjoint operator acting on a Hilbert space ${\mathcal H}$ and $P$ an orthogonal projection which may not commute with $H$.  Assume that $H^{1/2}P$ is densely defined, then the $P{\mathcal H}$-component $H_P\!\!\restriction\!_{P{\mathcal H}}$ of $H_P$ in \eqref{decomp-HP} satisfies $D[H_P\!\!\restriction\!_{P{\mathcal H}}] \subset D[H]$.
\end{proposition}
\noindent {\it Proof:} The form \eqref{quad-H1/2} restricted to $D[H^{1/2}P]$ becomes \eqref{quad-H1/2P}, so that $\|H^{1/2}w\|^2 = \|H^{1/2}Pw\|^2 = \|H_P^{1/2}w\|^2$ holds if $w \in D[H_P^{1/2}]$ proving thus the claim. \hfill
\qed

\begin{remark}
One can modify $H_P$ restricting it to the self-adjoint operator $H_P^\mathrm{mod}$ such that $H_P^\mathrm{mod}\!\!\restriction\!_{P{\mathcal H}} = H_P\!\!\restriction\!_{P{\mathcal H}}$ {\it and, at the same time,} $\,D[H_P^\mathrm{mod}] \subset D[H]$. To this end, it is enough to replace the zero operator in \eqref{decomp-HP} on $(I-P){\mathcal H}$ by its restriction to $(I-P)D[H]$, so that
\begin{align} \label{HPmd}
& H_P^\mathrm{mod}
  := H_P\!\!\restriction\!_{P{\mathcal H}}
                             \oplus 0\!\!\restriction\!_{(I-P)D[H]}\,,\\
& D[H_P^\mathrm{mod}]
  :=  D[H_P\!\!\restriction\!_{P{\mathcal H}}] \oplus (I-P)D[H]\,.\nonumber
\end{align}
Let us stress that in general $H$ is not an extension of $H_P^\mathrm{mod}$ either.
\end{remark}

Unless $H$ is bounded, the operator $H_P$ is generally different
from $PHP$. When $HP$ is densely defined, the symmetric operator
$PHP$ is not necessarily essentially self-adjoint on $D[HP]$, and
consequently, $H_P$ may not be the closure of $PHP$ either.

On the other hand, the quadratic form $u \mapsto \|H^{1/2}Pu\|^2$
defined on $D[H^{1/2}P]$ is a closed extension of the form $u
\mapsto \langle Pu,HPu \rangle$ defined on $D[HP]$, but in general
the former is not  the closure of the latter, because $D[HP]$ is not
necessarily dense in $D[H^{1/2}P]$. Indeed,  if $H$ is unbounded,
$D[H]$ is a proper subspace of $D[H^{1/2}]$. Take $u_0 \in D
[H^{1/2}]\backslash D[H]$ such that the vector $H^{1/2}u_0$ is
nonzero, and set $P$ to be the orthogonal projection onto the
one-dimensional subspace spanned by $u_0$. Taking into account that
$D[HP] = \{u \in {\mathcal H};\,\, Pu \in D[H]\}$ which $u_0 =Pu_0$
does not belong to, we find $HPu = 0$ for $u \in D[HP]$, while
$H^{1/2}Pu_0 = H^{1/2}u_0 \not= 0$ by assumption.


\section{Proof of Theorem~\ref{th:main'}}
\label{s: proof}

Most parts of the argument work for a general Hilbert space, however, there is a place where we have to assume ${\mathcal H}$ to be separable. It is clearly sufficient to prove formula \eqref{P(t)both} in Theorem~\ref{th:main'} because \eqref{P(t)right} and \eqref{P(t)left} easily follow from it, and furthermore, it is enough to consider $\varepsilon = 1$ and $t \geq 0$. For the sake of definiteness we use here and in the following the physicist convention about the inner product $\langle \cdot, \cdot \rangle$ on ${\mathcal H}$ supposing that it is antilinear in the \emph{first} argument.

\smallskip

The nonnegative self-adjoint operator $H$ we deal with can be conventionally represented through its spectral family $\{E(\lambda)\}_{\lambda\geq 0}$,
\begin{equation} \label{H-specrep}
 H = \int_{0-}^{\infty} \lambda E(\mathrm{d}\lambda)
\end{equation}
Given $\kappa>0$, we introduce
\begin{align}
K(\kappa) :&=\frac1{\kappa}[I-\mathrm{e}^{-i\kappa H}]
 = \frac{I-\cos \kappa H}{\kappa} +  i\frac{\sin \kappa H}{\kappa}
 \nonumber \\[-.5em]\nonumber\\[-.5em]
 &=: G(\kappa) + iH(\kappa)\,, \label{KGH}
\end{align}
where $G(\kappa)$ and $H(\kappa)$ are bounded self-adjoint
operators, and $G(\kappa)$ is nonnegative. Let $H^{\pm}(\kappa)$ be
nonnegative bounded self-adjoint operators which are the
nonnegative/negative parts of $H(\kappa)$, respectively. Then we
have
\begin{align}
|K(\kappa)| &= [G(\kappa)^2+H(\kappa)^2]^{1/2} = \tfrac{|\sin \frac12
\kappa H|}{\frac12\kappa}\,, \nonumber \\[.2em]
|I+ K(\kappa)| &=
\Big(I + (1+\kappa) \big(\tfrac{\sin \frac12\kappa H}{\frac12\kappa}\big)^2\Big)^{1/2}\,,
                                             \label{KGH2} \\[.5em]
 H(\kappa) &= H^+(\kappa) - H^-(\kappa)\,,
   \quad |H(\kappa)| = H^+(\kappa) + H^-(\kappa)\,. \nonumber
\end{align}
Using the spectral family $\{E(\lambda)\}_{\lambda\geq 0}$ from
\eqref{H-specrep}, we denote by $E_H^+(\kappa)$ the orthogonal projection
\begin{equation} 
E_H^+(\kappa) : \mathcal{H} \rightarrow \bigoplus_{m=0}^{\infty}
         E\big([\tfrac{2m\pi}{\kappa}, \tfrac{(2m+1)\pi}{\kappa})\big)
         \mathcal{H}
 \equiv E\Big(\bigcup_{m=0}^{\infty}
      [\tfrac{2m\pi}{\kappa},
      \tfrac{(2m+1)\pi}{\kappa})\Big)\mathcal{H}\,, \nonumber
\end{equation}
and put $E_H^-(\kappa) := I- E_H^+(\kappa)$. Note that
$E_H^+(\kappa) \overset s \to I$ and $E_H^-(\kappa) \overset s \to
0$ holds as $\kappa \rightarrow 0+$. Of course, both the
$E_H^{\pm}(\kappa)$ commute with $H(\kappa)$. As $E_H^+(\kappa)$ and
$E_H^-(\kappa)$ are nothing but the projections onto the closed
subspaces of $\mathcal{H}$ where $H(\kappa)$ becomes respectively a
nonnegative and negative self-adjoint operator,  we have
\begin{equation} \label{Hpm}
 H^+(\kappa) = H(\kappa)E_H^+(\kappa)\,,
       \quad H^-(\kappa) = -H(\kappa)E_H^-(\kappa)\,.
\end{equation}
We put $\kappa = |t|\tau$ with $\tau >0$, so that
$\kappa = \pm\, t\tau$ holds for $\pm\, t \geq 0$.

We will use the same notation as in \cite{EI05}, $F(\zeta;\tau) := P(\tau)\, \mathrm{e}^{-\zeta\tau H}P(\tau)$ and $S(\zeta;\tau) := \tau^{-1}[I-F(\zeta;\tau)]$ with $\text{\rm Im}\, \zeta \geq0$. These operator families are uniformly bounded and holomorphic for $\text{\rm Im}\, \zeta > 0$ which are the properties we used there. In this paper, however, we need only $\zeta= it$, in other words
\begin{align}
F(it;\tau) &= P(\tau)\, \mathrm{e}^{-it\tau H}P(\tau)\,,   \label{Ftau}\\
S(it;\tau) &= \tau^{-1}[I-F(it;\tau)]
  =\tau^{-1}[I- P(\tau)\,\mathrm{e}^{-it\tau H}P(\tau)]\,. \label{Stau}
\end{align}
It is easy to see that $F(it;\tau)$ in \eqref{Ftau} is a contraction and
$S(it;\tau)$ in \eqref{Stau} satisfies
\begin{align}
\mathrm{Re}\,(f,S(it;\tau)f)
& =\tau^{-1} \big[(f,f) -(f, P(\tau)\,\mathrm{e}^{-it\tau H}P(\tau)f) \big] \nonumber \\[.2em]
& \geq \tau^{-1} \big[\|f\|^2 -\|f\|\,\|P(\tau)\,\mathrm{e}^{-it\tau H}P(\tau)f\| \big] \nonumber \\[.2em]
& \geq \tau^{-1} \big[\|f\|^2- \|f\|^2 \big] =0 \nonumber
\end{align}
for all $f \in \mathcal{H}$, and consequently, $S(it;\tau)$ is an
$m$-accretive operator \cite{Ka76}. This means that $I+S(it;\tau)$ has a
bounded inverse and that $(I+ S(it;\tau))^{-1}$ is also a contraction.
The crucial observation is that to prove Theorem~\ref{th:main'} it is
sufficient to refer to Chernoff's result cited below,
and to verify that
\begin{equation} \label{Chern}
(I+S(it;\tau))^{-1} \overset s \longrightarrow (I+itH_P)^{-1}P
\quad \mathrm{as}\;\; \tau \rightarrow 0+
\end{equation}
holds for $t \in {\mathbb R}$.

\smallskip\noindent
\textbf{Chernoff's Theorem} (cf.~\cite[Theorem~1.1, pp.~4--6]{Ch74}, see also \cite{Ch68}):
\textsl{For a $t$-family $\{F(t)\}_{t\geq 0}$ of linear contractions on a Banach space and the generator $A$ of a strongly continuous contraction semigroup, the following two conditions are equivalent: }

\textrm{(a)} \textsl{For some $\lambda_0>0$, the family $\{\lambda_0 I + \tfrac{I-F(\varepsilon)}{\varepsilon}\}_{\varepsilon >0}$ converges strongly to $(\lambda_0 I+A)^{-1}$ as $\varepsilon \to 0+$. }

\textrm{(b)} \textsl{As $n\to \infty$, $\{F(\tfrac{t}{n})\}_{n=1}^{\infty}$ converges strongly to $\mathrm{e}^{tA}$, uniformly on bounded $t$-intervals.}

\medskip

Let us stress that our proof requires to demonstrate the convergence in \eqref{Chern} \emph{pointwise for any fixed $t$}; this will be sufficient to establish the convergence in the product formul{\ae} \eqref{P(t)both}--\eqref{P(t)left} as \emph{locally uniform} in $t \in {\mathbb R}$; in fact, we have only to deal with \eqref{P(t)both} as mentioned at the beginning of this section.

\begin{remark}
Since the function $F(it;\tau)$ in \eqref{Ftau} differs slightly from $F(t)$ appearing in condition (a) of Chernoff's Theorem, let us explain in detail how the product formula \eqref{P(t)both} follows from \eqref{Chern}, modifying to that purpose Chernoff's proof of the implication (a)$\,\Rightarrow\,$(b).

Consider the first the nontrivial part referring to the subspace $P{\mathcal H}$. Note that $S(it;\tau)$ generates a strongly continuous contraction semigroup $\{\mathrm{e}^{-\theta S(it;\tau)}\}_{\theta\geq 0}$ on ${\mathcal H}$, and the resolvent convergence \eqref{Chern} is equivalent to the convergence of the corresponding semigroups \cite[Theorem~IX.2.16]{Ka76}, hence for any $f\in P{\mathcal H}$ we have $\mathrm{e}^{-\theta S(it;\tau)}f \overset s \longrightarrow \mathrm{e}^{-i\theta tH_P}f$ as $\tau\to 0+$, uniformly on bounded intervals of the variable $\theta\geq 0$. In particular, choosing $\theta=1$ we get
\begin{equation} \label{remv-theta}
 \mathrm{e}^{-S(it;\tau)}f \longrightarrow \mathrm{e}^{-itH_P}f\quad\text{\rm as}\;\; \tau\to 0+\,
\end{equation}
for a fixed $t\geq 0$, and using the same equivalence in the opposite direction we infer that
$$
 (I+\lambda S(it;\tau))^{-1}f \longrightarrow
                        (I+i\lambda tH_P)^{-1}Pf \quad\text{\rm as}\;\;\tau\to 0+
$$
holds for any $\lambda\ge 0$ and $t\ge 0$. In particular, using the diagonal trick in the last relation with $\tau=1/n$ and $\lambda = 1/\sqrt{n}$, we obtain
\begin{equation} \label{diag-trick}
 (I +\tfrac1{\sqrt{n}}S(it; 1/n))^{-1}f \longrightarrow Pf
 \quad\text{\rm as}\;\; n \to \infty
\end{equation}
for every $t\geq 0$. Next we refer to \cite[Lemma 2]{Ch68} by which we have
$$
 \|[F(it; 1/n)^ng -\mathrm{e}^{-n(I- F(it; 1/n))}g\|
               \leq \sqrt{n}\,\|(I-F(it; 1/n))g\|
$$
for any $g \in {\mathcal H}$. Choosing $g= (I + \tfrac1{\sqrt{n}}S(it; 1/n))^{-1}f$ and having in mind that $f=Pf$, we conclude from here that
\begin{align*}
&\big\|[F(it; 1/n)^n - \mathrm{e}^{-S(it; 1/n)}]
                     (I + \tfrac1{\sqrt{n}}S(it; 1/n))^{-1}f\big\| \\
 &\qquad \leq \big\|(I + \tfrac1{\sqrt{n}}S(it; 1/n))^{-1}f - Pf\|\,,
\end{align*}
where by \eqref{diag-trick} the right-hand side tends to zero uniformly on bounded $t$-intervals as $n \to \infty$. Using the diagonal trick once again we get
$$
 \lim_{n\to\infty} \|F(it; 1/n)^nf - \mathrm{e}^{-S(it; 1/n)}f\| = 0\,
$$
uniformly on bounded $t$-intervals. Then the sought conclusion \eqref{P(t)both} follows immediately from relations \eqref{remv-theta} and \eqref{diag-trick}, since by \eqref{Ftau} we have $F(it; 1/n)^n =  [P(1/n) \mathrm{e}^{-itH/n} P(1/n)]^n$. Having dealt with the subspace $P\mathcal{H}$, the remaining case, $f \in (P{\mathcal H})^{\perp}$, is trivial: we have
$$
 [P(1/n) \mathrm{e}^{-itH/n} P(1/n)]^n f  \longrightarrow 0 \quad\text{\rm as}\;\; n \to \infty
$$
for each $t \geq 0$, since $P(1/n)f = P(1/n)(I-P)f$ tends by assumption to $P(I-P)f =0$ and $\mathrm{e}^{-itH_P}Pf =0$ holds at the same time.
\end{remark}

\smallskip

Let us add a remark on the conventions: here and in the following the convergence of operator families in the strong operator topology is denoted by $\overset s \longrightarrow$, and in the weak operator topology by $\overset w \longrightarrow$. A simple arrow is reserved for the convergence with respect to the norm of the Hilbert space $\mathcal{H}$, sometimes also dubbed `strong' -- we will occasionally use this term too -- while for the weak convergence in $\mathcal{H}$ we will again employ the symbol $\overset w \longrightarrow$. Later, in Proposition~\ref{p:sigmaweak-bdd}, we will use still another topology for convergence of families of Hilbert space vectors.

To proceed with the argument, we rewrite relation \eqref{Stau} as follows
\begin{align*}
S(it;\tau) &= \tau^{-1} \big[I- P(\tau)(\cos t\tau H- i\sin t\tau H)P(\tau)\big] \\[.2em]
&= \tfrac{I-P(\tau)}{\tau}
   + P(\tau)\,\tfrac{I-\cos t\tau H}{\tau}\,P(\tau)
  + i\,P(\tau)\,\tfrac{\sin t\tau H}{\tau}\,P(\tau) \\[.2em]
&= \tfrac{I-P(\tau)}{\tau} +P(\tau)tG(t\tau)P(\tau)
  + i\,P(\tau)tH(t\tau)P(\tau)\,,
\end{align*}
and consequently,
\begin{align} \label{IplusS}
I &+ S(it;\tau) \\[.2em]
&= I + \tau^{-1}(I-P(\tau))
    +P(\tau)tG(t\tau)P(\tau) +iP(\tau)tH(t\tau)P(\tau) \nonumber \\[.2em]
&= (1 +\tau^{-1})(I- P(\tau)) \oplus P(\tau)(I+ tG(t\tau)
+ itH(t\tau))P(\tau)\,, \nonumber
\end{align}
where $\oplus$ denotes the direct sum corresponding to the
decomposition of the Hilbert space into $P(\tau)\mathcal{H}$ and its
orthogonal complement.

To prove the sought relation \eqref{Chern} we need first a pair of lemmata.

\begin{lemma} \label{l:inverse}
The inverse of $I+S(it;\tau)$ in \eqref{IplusS} is given by
\begin{align} \label{inverse-of-IplusS}
(I &+ S(it;\tau))^{-1} \\
&= (1+\tau^{-1})^{-1}(I- P(\tau))
   \oplus \big[P(\tau)(I+ tG(t\tau) + itH(t\tau))P(\tau)\big]^{-1}. \nonumber
\end{align}
\end{lemma}

\vspace{-.3em}

\begin{proof}
The above expression and \eqref{IplusS} clearly multiply to
identity.
\end{proof}

\begin{lemma} \label{l:conv}
Let $u \in D[H^{1/2}]$, then in the limit $\kappa \rightarrow 0+$ we have
\begin{align*}
\mathrm{(i)}\,\, &G(\kappa)^{1/2} u \longrightarrow 0\,, \\[.2em]
\mathrm{(ii)}\,\,
 &H^+(\kappa)^{1/2} u \longrightarrow  H^{1/2}u\,, \;
  H^-(\kappa)^{1/2} u \longrightarrow  0\,, \; \mathrm{and}\;\;
  |H(\kappa)|^{1/2} u \longrightarrow  H^{1/2}u
\end{align*}
\end{lemma}
\begin{proof}
Using spectral theorem together with dominated convergence theorem we infer
that as $\kappa \rightarrow 0+$,
\begin{align*}
G(\kappa)^{1/2} u
 &= \int_{0-}^{\infty}
  \Big|\tfrac{\sin\tfrac{\kappa \lambda}2}{\tfrac{\kappa\lambda}2}\Big|^{1/2}
          |\sin\tfrac{\kappa \lambda}2|^{1/2}\,\lambda^{1/2}
   E(\mathrm{d}\lambda)u  \,\longrightarrow 0\,,
\end{align*}

\vspace{.5em}

$$
|H(\kappa)|^{1/2} u = \int_{0-}^{\infty} \Big|\tfrac{\sin \kappa
\lambda}{\kappa}\Big|^{1/2}
     E(\mathrm{d}\lambda)u \longrightarrow \int_{0-}^{\infty}
     \lambda^{1/2}  E(\mathrm{d}\lambda) u = H^{1/2}u
$$
and
\begin{align*}
H^+(\kappa)^{1/2} u
&= |H(\kappa)|^{1/2}E_H^+(\kappa) u\\[.2em]
&= \int_{0-}^{\infty} \Big|\tfrac{\sin \kappa
\lambda}{\kappa}\Big|^{1/2}
    \chi_{\bigcup_{m=0}^{\infty}
           [\tfrac{2m\pi}{\kappa}, \tfrac{(2m+1)\pi}{\kappa})}(\lambda)\,
           E(\mathrm{d}\lambda) u \\
 &\qquad\qquad
   \longrightarrow  \int_{0-}^{\infty} \lambda^{1/2}
     E(\mathrm{d}\lambda) u = H^{1/2}u
\end{align*}
This implies at the same time $H^-(\kappa)^{1/2} u \longrightarrow 0\,$ by \eqref{KGH2}.
\end{proof}

After these preliminaries, we are going to start the proof of
\eqref{Chern}. As we said in the introduction, the core of our
reasoning is the argument used by Kato \cite{Ka78}, see also
\cite[Supplements to Sect.~VIII.8]{RS80}.
In the same vein, with \eqref{Stau} in mind, we put for
$\tau>0$ and $t \in {\mathbb R}$
\begin{equation} \label{def-of-utau}
u_{\tau}(t) := (I+S(it;\tau))^{-1}f
\end{equation}
for an arbitrary but fixed $f \in \mathcal{H}$. We see from the expression
\eqref{inverse-of-IplusS} of $(I+S(it;\tau))^{-1}$ that
$u_{\tau}(t)$ is strongly continuous in $t$. At the same time, the
$\tau$-family $\{u_{\tau}(t)\} \subset {\mathcal H}$ with
$u_{\tau}(t)$ defined by
\eqref{def-of-utau} and $t \in {\mathbb R}$ fixed is uniformly bounded
by $\|f\|$, because $(I+S(it;\tau))^{-1}$ is a contraction.

Our next aim is to show that for each fixed $t \in {\mathbb R}$, the
family $\{u_{\tau}(t)\}$ converges in the Hilbert space norm to some
$u(t) \in {\mathcal H}$ as $\tau\to 0+$ and that $u(t) =
(I+itH_P)^{-1}Pf$. To achieve this goal we will have to analyze, in
particular, a new uniform property of the $\tau$-family
$\{u_{\tau}(t)\}$ at the final stage of the argument.

Using the identity \eqref{IplusS} we can invert relation
\eqref{def-of-utau} explicitly as
\begin{align} \label{fIplusS}
f&= (I+ S(it;\tau))u_{\tau}(t) \nonumber \\[.2em]
 &= u _{\tau}(t) + \tau^{-1}(I-P(\tau))u _{\tau}(t)
        +P(\tau)tG(t\tau)P(\tau)u_{\tau}(t) \nonumber\\
 &\qquad\qquad\qquad\qquad\qquad\qquad
         +iP(\tau)tH(t\tau) P(\tau)u _{\tau}(t) \nonumber \\[.2em]
 &= (1+ \tau^{-1})(I-P(\tau))u_{\tau}(t)
     \oplus P(\tau)[I+ tG(t\tau) +itH(t\tau)] P(\tau)u _{\tau}(t)\,.
\end{align}
Taking the inner product of $f$ with $(I-P(\tau))u _{\tau}(t)$ and
$P(\tau)u _{\tau}$, we get
\begin{equation} \label{innerIminusP}
\langle (I-P(\tau))u _{\tau}(t),f \rangle
= (1+\tau^{-1})\|(I-P(\tau))u _{\tau}(t)\|^2
\end{equation}
and
\begin{align} \label{innerP}
\langle P(\tau)u_{\tau}(t), f \rangle
&= \langle P(\tau)u_{\tau}(t), u_{\tau}(t) \rangle
  + \langle P(\tau)u_{\tau}(t), tG(t\tau)P(\tau)u_{\tau}(t)\rangle
                                                          \nonumber\\
 & \quad +i\langle P(\tau)u _{\tau}(t),
                 tH(t\tau) P(\tau)_{\tau}(t)\rangle \nonumber\\[.2em]
&= \|P(\tau)u _{\tau}(t)\|^2
  + \|(|t|G(t\tau))^{1/2}P(\tau)u_{\tau}(t)\|^2 \nonumber \\
  &\quad +i\big(\|(|t|H^+(t\tau))^{1/2} P(\tau)u_{\tau}(t) \|^2 \nonumber \\
  &\quad\quad -\|(|t|H^-(t\tau))^{1/2} P(\tau)u_{\tau}(t) \|^2\big)
\end{align}
for $\tau >0\,$; recall that $H(\kappa) = H^+(\kappa) - H^-(\kappa)$.

From the real part of \eqref{innerP} with \eqref{innerIminusP}, we infer,
using Schwarz inequality, that the $\tau$-families $\{P(\tau)u
_{\tau}(t)\}$ and $\{(I-P(\tau))u _{\tau}(t)\}$, as well as
$\{\tau^{-1}(I-P(\tau))u _{\tau}(t)\}$, are uniformly bounded by
$\|f\|$; the last claim naturally extends to
$\{\tau^{-1/2}(I-P(\tau))u _{\tau}(t)\}$ as long as $\tau\le 1$.
Moreover, using the real part of \eqref{innerP} again we can
conclude that the same is true for $\{(|t|G(\tau))^{1/2}
P(\tau)u_{\tau}(t)\}$. Let $f \in {\mathcal H}$ and $t \in {\mathbb
R}$ be arbitrary but fixed. It follows from the obtained uniform
boundedness that for each $t \in {\mathbb R}$, there exist a
(sub)sequence $\{\tau'\}_{0<\tau'\leq 1}$ of the family
$\{\tau\}_{0<\tau\le 1}$ with $\tau' \rightarrow 0+$ and vectors
$u(t),\, u_0(t)$ and $g(t)$ in ${\mathcal H}$ such that the
sequences $\{u_{\tau'}(t)\}$,
$\{(\tau')^{-1/2}(I-P(\tau'))u_{\tau'}(t)\}$ and $\{t^{1/2}
G(|t|\tau')^{1/2}u_{\tau'}(t)\}$ converge weakly to $u(t)$, $u_0(t)$
and $g(t)$, respectively, in ${\mathcal H}$ as $\tau' \to 0+$;
thus $\{P(\tau')u_{\tau'}(t)\}$ converges weakly to $Pu(t)$. We keep in mind
that, with the knowledge we have at the present stage, the limit
$u(t)$ may depend on the chosen sequence
$\{u_{\tau'}(t)\}_{0<\tau'\le 1}$.
\begin{lemma}\label{l:weaklimit}
One has
\begin{equation}\label{weaklimit}
 u(t) = Pu(t), \quad  u_0(t)=0, \quad g(t)=0,
\end{equation}
and therefore, as $\tau'\to 0+$,
\begin{align} \label{weakconv}
&u_{\tau'}(t) \overset w \longrightarrow  u(t) \,,\;\;
({\tau'})^{-1/2}(I-P(\tau'))u_{\tau'}(t)\overset w \longrightarrow  0\,,\;\; \\
& P(\tau')u _{\tau'}(t) \overset w \longrightarrow  Pu(t)\,,\;\;\nonumber
(|t|G(t\tau'))^{1/2} P(\tau')u_{\tau'}(t) \overset w \longrightarrow  0\,.
\end{align}
\end{lemma}
\noindent\emph{Proof.}
To begin with, the second limit in \eqref{weakconv} implies, in
particular, that $(I-P(\tau'))u _{\tau'}(t) \to 0$, hence we have
$(I-P)u(t) =0$ or $u(t) =Pu(t) $. Indeed, the uniform boundedness of
$\|(\tau')^{-1/2}(I-P(\tau'))u_{\tau'}(t)\|$ means that
$\|(I-P(\tau'))u_{\tau'}(t)\|\to 0$ which together with the weak convergence
implies $u(t) = Pu(t)$. Moreover, the same
uniform boundedness following from \eqref{innerIminusP} means that
$\tau^{-1}(I-P(\tau)) u_{\tau}(t)$ also converges weakly along
$\{\tau'\}$ which implies $u_0(t)=0$.

Secondly, for any fixed $w \in D[H^{1/2}]$ and $t \in {\mathbb R}$
we have the relation
\begin{align*}
\langle w, g(t) \rangle
& = \lim_{\tau'\to 0+}\, \langle w, (|t|G(t\tau'))^{1/2} P(\tau')u_{\tau'}(t)\rangle \\
& = \lim_{\tau'\to 0+}\, \langle (|t|G(t\tau'))^{1/2}w, P(\tau')u_{\tau'}(t)\rangle \\
& = \langle 0,Pu(t)  \rangle =0\,, 
\end{align*}
because by Lemma~\ref{l:conv} we have $(|t|G(t\tau'))^{1/2}w \rightarrow 0$
as $\tau'\rightarrow 0+$. This means that $g(t)=0$, because $D[H^{1/2}]$
is dense by assumption.
\qed

\medskip

Next we want to find out what one can deduce from the imaginary part of \eqref{innerP} about the properties of the operators $H^\pm(t\tau)$ introduced in \eqref{KGH2} and \eqref{Hpm}. To this aim we have to employ a topology different from those used up to now. Given a dense subspace $\mathcal{K}$ of $\mathcal{H}$, the symbol $\sigma(\mathcal{H}, \mathcal{K})$ denotes the \emph{weak topology on $\mathcal{H}$ defined by the dual pairing $\langle \mathcal{H}, \mathcal{K} \rangle$} or the \emph{$\mathcal{K}$-weak topology on $\mathcal{H}$} -- see, e.g., \cite[Sect.~IV.20.2]{Ko69} or \cite[Sect.~IV.5]{RS80}. In general, it is weaker (coarser) than the usual weak topology on $\mathcal{H}$, the latter being nothing else than
$\sigma(\mathcal{H}, \mathcal{H})$ in the just introduced notation

We see from the imaginary part of \eqref{innerP} that the difference between the respective elements of the $\tau$-families $\{\|(|t|H^+(t\tau))^{1/2} P(\tau)u_{\tau}(t) \|^2\}$ and $\{\|(|t|H^-(t\tau))^{1/2} P(\tau)u_{\tau}(t) \|^2\}$ is uniformly bounded, because by Schwarz inequality  the modulus of the left-hand side in
\begin{align*} 
&\mathrm{Im}\,\langle Pu(t),f \rangle \\
 & = \lim_{\tau\to 0+}\,\big[\|(|t|H^+(t\tau))^{1/2}P(\tau)u_{\tau}(t)\|^2
         -\|(|t|H^-(t\tau))^{1/2}P(\tau)u_{\tau}(t)\|^2\big]
\end{align*}
does not exceed $\|f\|^2$. This fact, unfortunately, does not tell us whether the two $\tau$-families of vectors, $\{(|t|H^{\pm}(t\tau))^{1/2} P(\tau)u_{\tau}(t)\}$, are separately
(uniformly) bounded, that is, whether each of them is (uniformly) weakly bounded. We have, however, at least the following result.
\begin{proposition} \label{p:sigmaweak-bdd}
Let $\{u _{\tau'}(t)\}$ be the subsequence appearing in Lemma~\ref{l:weaklimit}, weakly convergent to $u(t)$. Then the $\tau'$-families $\{(|t|H^{\pm}(t\tau'))^{1/2} P(\tau')u_{\tau'}(t)\}$ are \emph{Cauchy} sequences in the $\sigma(\mathcal{H}, D[H^{1/2}])$-weak topology, and as a result they are $\sigma(\mathcal{H}, D[H^{1/2}])$-weakly bounded. Furthermore, the family $\{(|t|H^{-}(t\tau'))^{1/2} P(\tau')u_{\tau'}(t)\}$ converges to zero in this topology as $\tau' \rightarrow 0+$.
\end{proposition}
\begin{proof}
Take an arbitrary $\phi \in D[H^{1/2}]$. We use Lemma~\ref{l:conv} which states, in particular, that $(|t|H^+(t\tau))^{1/2} \phi \rightarrow (|t|H)^{1/2}\phi$ holds when $\tau \rightarrow 0+$. In combination with \eqref{weakconv}, this yields
\begin{align*}
\langle \phi, (|t|H^+(t\tau'))^{1/2} P(\tau')u_{\tau'}(t)\rangle
&= \langle (|t|H^+(t\tau'))^{1/2} \phi, P(\tau') u_{\tau'}(t)\rangle
\nonumber\\
&\longrightarrow \langle (|t|H)^{1/2}\phi,  Pu(t)\rangle
\end{align*}
as $\tau' \rightarrow 0+$. For the minus sign, on the other hand, Lemma~\ref{l:conv} says that $(|t|H^-(t\tau))^{1/2}\phi \rightarrow 0$, and this implies
$$
\langle \phi, (|t|H^-(t\tau'))^{1/2}P(\tau')u_{\tau'}(t)\rangle
 \longrightarrow 0,  \quad \tau' \rightarrow 0+,
$$
because the self-adjointness of $(|t|H^-(t\tau'))^{1/2}$ in combination with Schwarz inequality gives
\begin{align*}
|\langle (|t|H^-(t\tau'))^{1/2}\phi, P(\tau')u_{\tau'}(t) \rangle|
&\leq \|(|t|H^-(t\tau'))^{1/2}\phi\|\,\|P(\tau')u_{\tau'}(t) \| \\
&\leq \|(|t|H^-(t\tau'))^{1/2}\phi\|\,\|f\|\
\longrightarrow 0\,.
\end{align*}
This yields the stated assertion, including the fact that the family $\{(|t|H^-(t\tau'))^{1/2}P(\tau')u_{\tau'}(t)\}$ is $\sigma(\mathcal{H}, D[H^{1/2}])$-weakly convergent to zero.
\end{proof}

\medskip


On the other hand, it is not clear whether $Pu(t)$ belongs to $D[H^{1/2}]$, that is, whether the plus-sign family $\{(|t|H^{+}(t\tau))^{1/2} P(\tau)u_{\tau}(t)\}$ converges to $H^{1/2}Pu(t)$ in the $\sigma(\mathcal{H}, D[H^{1/2}])$ topology. What is important, the information we were able to deduce in this way about the convergence of the families $\{(|t|H^\pm(t\tau'))^{1/2} P(\tau')u_{\tau'}(t)\}$ is too limited; it does not seem possible to apply the same procedure as we used in Lemma~\ref{l:weaklimit} for $\{(|t|G(t\tau'))^{1/2} P(\tau')u_{\tau'}(t)\}$.

This forces us to seek a different strategy for the proof of Theorem~\ref{th:main'} that would allow us to identify the vector $u(t)$ with $(I +itH_P)^{-1}Pf$, in other words, to demonstrate relation \eqref{Chern} claiming that for every fixed $t\in {\mathbb R}$, the family $\{(I+S(it;\tau))^{-1}f\}$, or otherwise $\{u_{\tau}(t)\}$ in accordance with \eqref{def-of-utau}, converges to $(I+itH_P)^{-1}Pf$ in the Hilbert space norm as $\tau \to 0+$.

The argument is somewhat subtle and relies on our previous work \cite{EI05}, see also \cite{EINZ07}, about the Zeno product formul{\ae} related to
Theorems~\ref{th:main} and \ref{th:main'}. In those papers we demonstrated that the family $\{u_{\tau}(t)\}$ has a unique limit, namely
$(I+itH_P)^{-1}Pf$, as $\tau \to 0+$. As we mentioned in the introduction, however, the obtained convergence referred neither to the norm of ${\mathcal H}$ nor even to the weak topology. Precisely speaking, it is shown in \cite{EI05} that \eqref{Pboth}--\eqref{Pleft} and \eqref{P(t)both}--\eqref{P(t)left} hold in the topology of the Fr\'echet space $L^2_{\text{\rm loc}}({\mathbb R}; {\mathcal H}) = L^2_{\text{\rm loc}}({\mathbb R})\otimes {\mathcal H}$ of the ${\mathcal H}$-valued strongly measurable functions $v(\cdot)$ on ${\mathbb R}$ such that the $\|v(\cdot)\|$ are locally square integrable there, equipped with the topology induced by the family of semi-norms $v \mapsto \big(\int_a^b \|v(t)\|^2\, \mathrm{d}t\big)^{1/2}$ for any bounded interval $(a,b)$ with $a<b$. This follows from \cite[Lemma 3.1, p.~200]{EI05} which says that for every bounded closed interval $[a,b] \subset {\mathbb R}$ one has
\begin{equation} \label{L2locCov}
 \int_a^b \|u_{\tau}(t) - (I+ itH_P)^{-1}Pf\|^2\, \mathrm{d}t
   \; \longrightarrow \, 0 \quad\mathrm{as} \;\: \tau \to 0+\,;
\end{equation}
the reader may compare this result with \eqref{Chern}.

Note also that \eqref{L2locCov} implies that for every
$f \in {\mathcal H}$, there exist a set
$M_f \subset {\mathbb R}$ of Lebesgue measure zero, possibly dependent
on $f$, and a (sub)sequence $\{\tau'\}_{0<\tau' \leq 1}$ of
$\{\tau\}_{0<\tau \leq 1}$ along which it holds that for all $s \in
{\mathbb R}\setminus M_f$,
\begin{equation}\label{a.e.Cov}
  u_{\tau}(s) \longrightarrow  (I+ isH_P)^{-1}Pf
                   \quad\text{\rm in\;the\;norm\;of}\;\: {\mathcal H},
\end{equation}
in other words, $u_{\tau'}(s) \longrightarrow  (I+ isH_P)^{-1}Pf$ as $\tau'\to 0$. With the coming argument in mind, it is useful to note that the set $\mathbb {R} \setminus M_f$ at which the convergence takes place is dense in ${\mathbb R}$. Furthermore, since ${\mathcal H}$ is separable by assumption, we can choose a countable dense subset ${\mathcal D} := \{f_l\}_{l=1}^{\infty}$ in  ${\mathcal H}$. Putting $M = M_{\mathcal D} := \cup_{l=1}^{\infty} M_{f_l}$, which is also a set of Lebesgue measure zero, we may then say that \eqref{a.e.Cov} holds for all $s \in {\mathbb R}\setminus M$ and for every $f \in {\mathcal D}$, and hence, in view of the density, also for every $f \in {\mathcal H}$.

Moreover, we note that $s=0$ does not belong to $M_f$ for any $f \in {\mathcal H}$, and therefore it neither belongs to $M$. Indeed, using Lemma~\ref{l:inverse} and \eqref{inverse-of-IplusS} in combination with the continuity of $\tau \mapsto P(\tau)$ it is easy to see that
\begin{equation}\label{Cov-at-0}
 u_{\tau}(0) = (I+S(0;\tau))^{-1}f
             = (I+\tau^{-1})^{-1}(I-P(\tau)) \oplus P(\tau)f
             \longrightarrow Pf,
\end{equation}
which means $0 \notin M_f$.

\medskip

Now let us first briefly outline our plan of \emph{how to complete the proof of Theorem~\ref{th:main'}}; we recall that one has to verify relation \eqref{Chern} showing that for every $t\in {\mathbb R}$, $\{u_{\tau}(t)=(I+S(it;\tau))^{-1}f\}$ converges to $(I+itH_P)^{-1}Pf$ in the Hilbert space norm for all $t$ as $\tau \to 0+$. Let us recall here that, as already mentioned in the text following \eqref{Chern}, the use of Chernoff's theorem only requires to establish the convergence in \eqref{Chern} pointwise for all $t \in {\mathbb R}$. Here and in the following we keep in mind that $\{u_{\tau}(t)\}_{0<\tau\leq 1}$ is uniformly bounded in both $t$ and $0<\tau\leq 1$, i.e.
\begin{equation}  \label{utau:bdd-tau-t}
\sup_{0<\tau\leq 1}\sup_{t\in {\mathbb R}} \|u_{\tau}(t)\| \leq \|f\|,
\end{equation}
since $(I+S(it;\tau))^{-1}$ is a contraction. We need not strive to show its local uniformity, however, the following reasoning will establish this property with respect to $t \in {\mathbb R} \setminus \{0\} $. We will proceed in four steps demonstrating validity of the following claims:

\smallskip

I. The $\tau$-family $\{u_{\tau}(t)\}_{0<\tau \leq 1}$ of vectors
$u_{\tau}: {\mathbb R} \ni t \mapsto u_{\tau}(t) \in {\mathcal H}$ is
equicontinuous in $t \in {\mathbb R} \setminus \{0\}$ with respect to
the strong topology (i.e., Hilbert space norm) of ${\mathcal H}$.

\smallskip
II. The family $\{u_{\tau}(t)\}_{0<\tau\leq 1}$ converges as $\tau\to 0+$ for each fixed $t \in {\mathbb R}$ to some $u(t) \in{\mathcal H}$ in the weak topology of ${\mathcal H}$, and furthermore, the convergence is  \emph{even} locally uniform with respect to $t \in {\mathbb R} \setminus \{0\}$. The limit function $u(t)$ turns out to be continuous in $t \in {\mathbb R}$ in the weak topology of ${\mathcal H}$.

\smallskip
III. The limit satisfies
\begin{equation} \label{u(t)=}
 u(t) = (I+ itH_P)^{-1}Pf \quad \text{for\,\, all}\,\, t.
\end{equation}

\smallskip
IV. Finally, the family $\{u_{\tau}(t)\}_{0<\tau\leq 1}$ converges as $\tau\to 0+$ for any fixed $t \in {\mathbb R}$ to $u(t) \equiv (I+ itH_P)^{-1}Pf$ in the strong topology of ${\mathcal H}$, and furthermore, the convergence is \emph{even} locally uniform with respect to $t \in {\mathbb R} \setminus \{0\}$.

\medskip

\emph{Step I.} The claim is expressed in the following lemma about the
local equicontinuity of the $\tau$-family $\{u_{\tau}(t)\}$
in a neighbourhood of $t=s$ in ${\mathbb R}\setminus \{0\}$. It plays
a crucial role in concluding the proof of Theorem~\ref{th:main'}.
\begin{lemma} \label{l:equiconti}
Let $f \in {\mathcal H}$.
Then the $\tau$-family $\{u_{\tau}(t)\}_{0<\tau \leq 1}$ of vectors
$u_{\tau}: {\mathbb R}\setminus \{0\} \ni t \mapsto u_{\tau}(t) \in {\mathcal H}$
is equicontinuous locally in $t$ with respect to the strong topology
on ${\mathcal H}$. More explicitly, for every $\varepsilon>0$ and for
every $s \in {\mathbb R}\setminus \{0\}$ there exists an
$s$-dependent constant $\delta=\delta(f;\varepsilon;s)>0$ such that if
$t,s > 0$ or $\:t,s< 0$ with $|t-s|<\delta$, then
$\|u_{\tau}(t)-u_{\tau}(s)\| < \varepsilon$ holds for all $0<\tau \leq 1$.
\end{lemma}

\smallskip

We postpone for the moment the proof of Lemma~\ref{l:equiconti}, returning
to it in Sect.~\ref{s: proof-lemma}, and accept its claim, to finish
first the proof of Theorem~\ref{th:main'}.

\medskip

\emph{Step II.}
Without loss of generality we may suppose $f \not= 0$. If $t=0$, the family $\{u_{\tau}(0)\}_{0<\tau\leq 1}$ converges strongly to $Pf$ as $\tau\to 0+\,$ as mentioned above, cf. \eqref{Cov-at-0}. Consider thus a nonzero $t \in {\mathbb R}\setminus \{0\}$ and take \emph{arbitrary} (sub)se\-quen\-ce $\{\tau'\}_{0<\tau' \leq 1}$ of $\{\tau\}_{0<\tau \leq 1}$ with $\tau'\to 0+$. By Lemma~\ref{l:equiconti}, to be yet proven, we see the $\tau'$-family $\{u_{\tau'}(t)\}_{0<\tau' \leq 1}$ is equicontinuous in the strong topology and therefore in the weak topology, because the `full' $\tau$-family $\{u_{\tau}(t)\}_{0<\tau\leq 1}$ is. We observed in \eqref{utau:bdd-tau-t} that the vectors $u_{\tau}(t)$ with any $t \in {\mathbb R}\setminus \{0\}$ and $\tau\in(0,1]$ lie in the closed ball $\bar{B}(0;\|f\|) \subset {\mathcal H}$ with the center at the origin and radius $\|f\|$, which is weakly compact. To proceed, note that $\bar{B}(0;\|f\|)$ is {\it metrizable in the weak topology}, since ${\mathcal H}$ is separable by assumption, see e.g. \cite[Problem/Solution 18, p.~12 and 181]{Ha67}. Thus the equicontinuity holds with respect to the metric on the space $\bar{B}(0;\|f\|)$ equivalent to the weak topology on it. Then, by virtue of the Ascoli--Arzel\`a theorem,  see e.g. \cite[p.~81]{KeNa76} or \cite[Thm.~1.5.3]{Si15}, there exists a (sub)sequence $\{\tau^{\prime\prime}\}_{0<\tau^{\prime\prime}\leq 1}$ of $\{\tau'\}_{0<\tau' \leq 1}$ with $\tau^{\prime\prime} \to 0+$, along which $\{u_{\tau}(t)\}$ converges \emph{weakly} to some limit $u(t)$, and moreover, the convergence is \emph{(locally) uniform} in $t \in {\mathbb R}\setminus \{0\}$. This means that the numerical family $\{\langle \psi, u_{\tau^{\prime\prime}}(t) \rangle\}$ converges for every fixed $\psi \in {\mathcal H}$ as $\tau^{\prime\prime} \to 0+$, (locally) uniformly in $t \in {\mathbb R}\setminus \{0\}$. Putting this together with the case $t=0$ mentioned above, we see that the limit $\langle \psi, u(t)\rangle$ is $t$-continuous everywhere in ${\mathbb R}$, in other words, that $u(t)$ is $t$-continuous everywhere in ${\mathbb R}$ with respect to the weak topology of ${\mathcal H}$.

\medskip

\emph{Step III.} We have established above two ways of convergence of the original $\tau$-family $\{u_{\tau}(t)\}_{0<\tau \leq 1}$ of ${\mathcal H}$-valued $t$-continuous functions as $\tau\to 0+$, or more specifically, convergence with respect to two different topologies. One is the convergence to $u(t)$ in the weak topology of ${\mathcal H}$, (locally) uniformly for $t \in {\mathbb R}\setminus \{0\}$, and the other is the convergence to $(I+itH_P)^{-1}Pf$ in $L^2_{\text{\rm loc}}({\mathbb R}; {\mathcal H})$ in the strong, and therefore also weak sense. This allows us to conclude that these two limit vectors coincide for all $t$, in other words, to establish the equality \eqref{u(t)=}. Indeed, for any $\psi \in {\mathcal H}$ and an arbitrary bounded interval $(a,b)$ with $a<b$ we have the following elementary estimate,
\begin{align*}
&\int_a^b |\langle \psi, u(t)- (I+itH_P)^{-1}Pf\rangle|\, \mathrm{d}t \\
&\leq \int_a^b |\langle \psi, u(t)- u_{\tau}(t)\rangle|\, \mathrm{d}t
  +\int_a^b |\langle \psi, u_{\tau}(t)- (I+itH_P)^{-1}Pf\rangle|\, \mathrm{d}t\,,
\end{align*}
the last term of which tends to zero as $\tau\to 0+$ by \eqref{L2locCov}. At the same time, the first term on the right-hand side also tends to zero in view of the weak convergence and the dominated convergence theorem. Consequently the integrated expression on the left-hand side vanishes identically,
$$
 \langle \psi, u(t)- (I+itH_P)^{-1}Pf\rangle = 0,
$$
for all $t$ and for all $\psi \in {\mathcal H}$, which yields the desired claim \eqref{u(t)=}.
In this way, we have obtained the relations
\begin{align}
 &u(t)  =Pu(t)  \in D[H_P] = D[(H^{1/2} P)^*(H^{1/2} P)]\,,\nonumber\\[-.6em]
                                                 \label{limits2}\\[-.6em]
 &Pf= P(I+ itH_P)u(t) = P[I+it(H^{1/2} P)^*(H^{1/2} P)]u(t)\,, \nonumber
\end{align}
which show that $\{u_{\tau}(t)\}$ converges to $u(t)= (I+ itH_P)^{-1}Pf$ in the weak topology of ${\mathcal H}$
along $\tau \to 0+$, and therefore also along $\tau\to 0+$.

\medskip

\emph{Step IV.} It remains to demonstrate that $\{u_{\tau}(t)\}$ converges to $u(t)= (I+ itH_P)^{-1}Pf$ also \emph{in the Hilbert space norm} for all $t \in {\mathbb R}$. For the sake of completeness we shall show that the same claim can also be made about the $\tau$-family $\{(tG(\tau))^{1/2}P(\tau)u_{\tau}(t)\}$ in \eqref{innerP}; note that the analogous question concerning $\{(|t|H^{\pm}(t\tau))^{1/2}P(\tau) u_{\tau}(t)\}$ also appearing in \eqref{innerP} is more complicated and we are able to provide only a partial answer to it, cf. Proposition~\ref{p:sigmaweak-bdd}.

Since we have already established the weak convergence of the $\tau$-family
$\{u_{\tau}(t)\}$ for all $t$, we need only to show that the $\tau$-families
of the norms of these vectors converge. To this end, we observe again the real
part of \eqref{innerP}, however, replacing now $f$ by $P(\tau)f$ in the inner
product; we write also the imaginary part for purpose of a further discussion.
Writing the left-hand side of \eqref{innerP} as
$\langle u_{\tau}(t),P(\tau)f \rangle$, we obtain for its real and imaginary
part
\begin{align}
\mathrm{Re}\,\langle u_{\tau}(t),P(\tau)f \rangle
&= \|P(\tau)u_{\tau}(t)\|^2 +\|(|t|G(t\tau))^{1/2}P(\tau)u_{\tau}(t)\|^2, \label{real} \\[.5em]
\mathrm{Im}\,\langle u_{\tau}(t),P(\tau)f \rangle
&=  \|(|t|H^+(t\tau))^{1/2}P(\tau)u_{\tau}(t)\|^2 \nonumber\\
    &\qquad\qquad\qquad
                     - \|(|t|H^-(t\tau))^{1/2}P(\tau)u_{\tau}(t)\|^2.
                                                             \label{imaginary}
\end{align}
In view of \eqref{limits2}, the continuity of the projection family $\{P(\tau)\}$, and the weak convergence of $\{u_{\tau}(t)\}$ as $\tau\rightarrow 0+$, which we have already established, the left-hand sides of the last two relations converge, leading to the following limits
\begin{align}
\text{\rm Re}\,\langle u_{\tau}(t),P(\tau)f \rangle
 &\longrightarrow \langle u(t),Pu(t) \rangle = \|Pu(t)\|^2.
                                             \label{real2} \\[.5em]
\text{\rm Im}\,\langle u_{\tau}(t),P(\tau)f \rangle
 &\longrightarrow \langle u(t),tH_Pu(t) \rangle
                     =  \|(|t|H)^{1/2}Pu(t)\|^2\label{imaginary2}
\end{align}
valid for \emph{all} $t$. Comparing now the right-hand sides of relations \eqref{real} and \eqref{real2} we get
\begin{align*}
 \|Pu(t)\|^2 &= \lim_{\tau \to 0+}\, \big[\|P(\tau)u_{\tau}(t)\|^2
         +\|(|t|G(t\tau))^{1/2}P(\tau)u_{\tau}(t)\|^2 \big]\\
 &= \liminf_{\tau \to 0+}\, \big[\|P(\tau)u_{\tau}(t)\|^2
         + \|(|t|G(t\tau))^{1/2}P(\tau)u_{\tau}(t)\|^2 \big] \\
 &\geq \liminf_{\tau \to 0+}\, \|P(\tau)u_{\tau}(t)\|^2
         +\liminf_{\tau \to 0+}\,  \|(|t|G(\tau))^{1/2}P(\tau)u_{\tau}(t)\|^2 \\
 &\geq \|Pu(t)\|^2 + \|0\|^2 =  \|Pu(t)\|^2,
\end{align*}
which means that the two terms on the right-hand side of \eqref{real} converge
to $\|Pu(t)\|^2$ and to zero, respectively, yielding thus the sought
convergence in the Hilbert space norm along the sequence $\tau \to 0+$,
\begin{equation} \label{strongconv}
  P(\tau)u_{\tau}(t) \longrightarrow Pu(t)\,, \quad
  (|t|G(t\tau))^{1/2}P(\tau)u_{\tau}(t) \longrightarrow 0\,.
\end{equation}
The first convergence in \eqref{strongconv} also implies
$u_{\tau}(t) \rightarrow u(t) = Pu(t)$, i.e. convergence of
$\tau$-family$\{u_{\tau}(t)\}$ to $u(t)=Pu(t)$ in the Hilbert space
norm, because we know already that $(I-P(\tau'))u _{\tau'}(t) \to
0$.

This shows  nothing but our sought statement of the convergence \eqref{Chern} for the family $\{u_{\tau}(t)\}$ which, with the reference to Chernoff's criterion \cite{Ch68, Ch74}, concludes the proof of Theorem~\ref{th:main'}.
\qed

\medskip
The proof of Lemma~\ref{l:equiconti} will be given in the next section.


\section{Proof of Lemma~\ref{l:equiconti}}
\label{s: proof-lemma}

Recall first that the Hilbert space ${\mathcal H}$ we are handling is assumed to be separable which is the property we needed to construct the exceptional set $M$, cf. the text following \eqref{a.e.Cov}. Recall also
the notation we introduced in \eqref{def-of-utau} for the vector obtained by application of the operator $(I+ S(it;\tau))^{-1}$ to an arbitrary $f \in {\mathcal H}$,
\begin{align}
u_{\tau}&(t)
 = (I+ S(it;\tau))^{-1}f             \nonumber\\
&= \big\{(1+\tau^{-1})^{-1}(I-P(\tau))
 \oplus
  [P(\tau)(I+ tG(t\tau) + itH(t\tau))P(\tau)]^{-1}\big\}f \nonumber\\
&=\big\{(1+\tau^{-1})^{-1}(I-P(\tau))\oplus T^P(t;\tau)\big\}f\,,
\end{align}
where the operator
\begin{equation} \label{def-TP(t;tau)}
T^P(t;\tau) := [P(\tau)(I+ tG(t\tau) +itH(t\tau))P(\tau)]^{-1}
\end{equation}
may be also considered on the whole Hilbert space ${\mathcal H}$,
although in the proper sense it is an operator with the domain and range
included in the closed subspace $P(\tau){\mathcal H}$; one may regard it
as vanishing on the orthogonal complement $(I-P(\tau)){\mathcal H}$.

Next, for fixed $t,s >0$ or $t,s <0$, we put
\begin{align} \label{D(t,s;tau)}
D(t,s;\tau)f
:&= u_{\tau}(t) - u_{\tau}(s)                         \nonumber\\
&= [(I+ S(it;\tau))^{-1}- (I+ S(is;\tau))^{-1}]f      \nonumber\\
&= (I+ \tau^{-1})^{-1}(I-P(\tau))
      \oplus \big[T^P(t;\tau) -T^P(s;\tau)\big]f\,.
\end{align}
With the above direct sum decomposition in mind, one may for a fixed $\tau$ deal with the operator on the subspace $P(\tau){\mathcal H}$,
\begin{equation} \label{TP}
  T^P(t,s;\tau) := D(t,s;\tau)\restriction_{P(\tau){\mathcal H}}
                 = T^P(t;\tau) -T^P(s;\tau)\,,
\end{equation}
however, since $\tau$ is varying, we have to consider it on the whole space ${\mathcal H}$. Note that the $\tau$-family $\{D(t,s;\tau)\}_{0<\tau \leq 1}$ is
strongly continuous in $0<\tau \leq 1$, and uniformly bounded, i.e. $\|D(t,s;\tau)\| \leq 2$,  because both $(I+ S(it;\tau))^{-1}$ and $(I+ S(is;\tau))^{-1}$ are contractions. Note also that we have $ D(t,s;\tau) = P(\tau)D(t,s;\tau)$.

\medskip

To verify the assertion of Lemma~\ref{l:equiconti}, we need to show that for any fixed $f \in {\mathcal H}$, the difference $D(t,s;\tau)f$ in \eqref{D(t,s;tau)} with $t,s \in {\mathbb R} \setminus \{0\}$ converges to zero in the Hilbert space norm as $|t-s| \to 0$, and that the convergence is uniform with respect to $\tau\in(0,1]$.

We use a small trick showing first that it is sufficient to establish the claim of Lemma~\ref{l:equiconti} \emph{under the additional assumption} that one of the $t$ and $s$, say the latter, belongs to ${\mathbb R} \setminus (M\cup \{0\})$. Indeed, if this is the case, i.e. if for any fixed $s\in{\mathbb R} \setminus (M\cup \{0\})$ and arbitrary $f \in {\mathcal H},\;\varepsilon>0$, there is a
$\delta =\delta(f;\varepsilon;s)>0$ such that
$$
 \|u_{\tau}(t) -u_{\tau}(s)\| < \tfrac{\varepsilon}2 \;\;
\text{\rm holds for all}\;\; t \in (s-\delta,s+\delta)\;\; \text{\rm and}\;\; \tau\in (0,1],
$$
the lemma is valid in the general case as well. To see that, we take any two points $t_1,\,t_2  \in {\mathbb R} \setminus \{0\}$ in the vicinity of the chosen $s$ satisfying $|t_i-s|<\tfrac{\delta}2$ for $i=1,2$, then
\begin{align*}
&|t_1-t_2| \leq |t_1-s|+|s-t_2| < \tfrac{\delta}2+ \tfrac{\delta}2=\delta, \\
&\|u_{\tau}(t_1) -u_{\tau}(t_2)\|
\leq \|u_{\tau}(t_1) -u_{\tau}(s)\| +\|u_{\tau}(s) -u_{\tau}(t_2)\|
< \tfrac{\varepsilon}2 +\tfrac{\varepsilon}2 = \varepsilon
\end{align*}
holds independently of $\tau$. This yields the `full' claim of Lemma~\ref{l:equiconti} in view of the fact that the set ${\mathbb R} \setminus (M\cup \{0\})$ from which the number $s$ is chosen is dense in ${\mathbb R} \setminus \{0\}$.

\smallskip

Let us thus turn to the nontrivial part of the proof which consists of establishing the assertion of Lemma~\ref{l:equiconti} for $s \in {\mathbb R} \setminus (M\cup \{0\})$ and $t \in {\mathbb R} \setminus \{0\}$; the aim is to complete the proof of our main result in the forthcoming Lemma~\ref{l:estimate-T}(ii). We will work out the argument assuming that $t,s > 0$, the case with the opposite signs is completely analogous.

To begin with, we use functional calculus with \eqref{KGH}, \eqref{KGH2}
to rewrite the differences containing $H(t\tau)$ and $G(t\tau)$
in \eqref{D(t,s;tau)} as
\begin{align} %
\hspace{-1em}tH(t\tau) - sH(s\tau)
&= \tfrac{\sin t\tau H - \sin s\tau H}{\tau}
 = \tfrac{2}{\tau} \cos(\tfrac{t+s}2\tau H)\, \sin(\tfrac{t-s}2\tau H),
                        \label{tH(t)-sH(s)}\\
tG(t\tau) - sG(s\tau)
&= \tfrac{-\cos t\tau H + \cos s\tau H}{\tau}
= \tfrac{2}{\tau} \sin(\tfrac{t+s}2\tau H)\, \sin(\tfrac{t-s}2\tau H).
                       \label{tG(t)-sG(s)}
\end{align}
To simplify expressions in the discussion to follow, we recall the quantity
$K(\kappa)$ introduced in \eqref{KGH} which for $\kappa = t\tau$ allows us
to write
\begin{equation}
I + tK(t\tau) = I + tG(t\tau) +itH(s\tau)
= I + \tfrac{I-\cos t\tau H}{\tau} + i\tfrac{\sin t\tau H}{\tau}\,,
                                                  \label{KGH3}
\end{equation}
and similarly for $k=s\tau$. Furthermore, for a given $\tau\in (0,1]$ we
intro\-duce the self-adjoint operator
\begin{equation} \label{Htau}
     H_{\tau} := H(I +\tau H)^{-1}
\end{equation}
which is positive and bounded, and note that $I+|s|H_\tau$ has a bounded
inverse for any $s\in\mathbb{R}$. The difference $D(t,s;\tau)$ in
\eqref{D(t,s;tau)} which can be identified with its nontrivial part
$T^P(t,s;\tau)$ in \eqref{TP} can be then rewritten as
\begin{align}
D(t,s;\tau) &= T^P(t,s;\tau)                                  \nonumber\\
&= [P(\tau)(I+ tK(t\tau))P(\tau)]^{-1}
        - [P(\tau)(I+ sK(s\tau))P(\tau)]^{-1}               \nonumber\\
&=  [P(\tau)(I+ tK(t\tau))P(\tau)]^{-1}                     \nonumber\\
  &\qquad \cdot \big\{P(\tau)(I+ sK(s\tau))P(\tau)
        -P(\tau)(I+ tK(t\tau))P(\tau)\big\}                 \nonumber\\
  &\quad\qquad \cdot [P(\tau)(I+ sK(s\tau))P(\tau)]^{-1}          \nonumber\\
&= [P(\tau)(I+ tK(t\tau))P(\tau)]^{-1}
    [P(\tau)(sK(s\tau) -tK(t\tau))P(\tau)]                  \nonumber\\
  &\qquad \cdot [P(\tau)(I+ sK(s\tau))P(\tau)]^{-1}          \nonumber\\
&= [P(\tau)(I+ tK(t\tau))P(\tau)]^{-1}                       \nonumber\\
  &\qquad \cdot [(sK(s\tau) -tK(t\tau))(I+ |s|H_{\tau})^{-1}]  \nonumber\\
  &\quad\qquad \cdot (I+ |s|H_{\tau}) P(\tau)[P(\tau)(I+ sK(s\tau))P(\tau)]^{-1}                                                                \nonumber\\
&=: T_1(t;\tau)\, T_2(t,s;\tau)\, T_3(s;\tau),
                                         \label{split-D(t,s;tau)to123}
\end{align}
where we have introduced
\begin{align}
&T_1(t;\tau) = T^P(t;\tau) = [P(\tau)(I+ tK(t\tau))P(\tau)]^{-1},
                                                           \label{T1tau} \\
&T_2(t,s;\tau) = (sK(s\tau) -tK(t\tau))(I+ |s|H_{\tau})^{-1},\label{T2tau} \\
&T_3(s;\tau)
 =(I+ |s|H_{\tau})P(\tau)[P(\tau)(I+ sK(s,\tau))P(\tau)]^{-1},\label{T3tau}
\end{align}
which are all bounded operators on $P(\tau){\mathcal H}$ as well as on ${\mathcal H}$; in \eqref{T2tau} we are able to drop the projections $P(\tau)$ appearing in the last formula since the operator \eqref{def-TP(t;tau)} maps the subspace $P(\tau)\mathcal{H}$ onto itself.

Next we consider the following two operator-valued functions, which may be thought of as the $\tau$-limits of the families $\{T_1(t;\tau)\}_{0<\tau\leq 1}\,$ in \eqref{T1tau}, and $\{T_3(s;\tau)\}_{0<\tau\leq 1}\,$ in \eqref{T3tau}, namely
\begin{align}
T_1(t) :&= [(I+ itH_P)^{-1}P],            \label{T1} \\
T_3(s) :&= (I+ |s|H)P[(I+ isH_P)^{-1}P]
         = (I+ |s|H)[(I+ isH_P)^{-1}P],          \label{T3}
\end{align}
where we can remove again one $P$ from the second expression of \eqref{T3} since $(I+ isH_P)^{-1}$ maps $P{\mathcal H}$ to itself. We already know from \eqref{strongconv} that $T_1(t)$ is the strong limit as $\tau\to 0+$ of the family $\{T_1(t;\tau)\}$ of contractions, for the moment at least as long as $t \in {\mathbb R}\setminus (M \cup \{0\})$.

Next we are going to show that $T_3(s)$ can be extended to a bounded operator on ${\mathcal H}$. We begin with a crucial observation.

\begin{lemma} \label{l:T0}
Let $H$ be our nonnegative self-adjoint operator acting in ${\mathcal H}$ and $H_P$ the self-adjoint operator introduced in Sect. 2 referring to the orthogonal projection $P$. Consider the operator
\begin{equation} \label{def-T0}
       T_0 := (I+H)[(I+ H_P)^{-1}P]
\end{equation}
in ${\mathcal H}$ for which we have: \text{\rm (i)} the domain and range of $T_0$ are
\begin{align}
 D[T_0] &= P(I+H_P) D[HP]
         = P(I+H_P) (D[H]\cap P{\mathcal H}) \oplus (P{\mathcal H})^{\perp}, \nonumber \\
 R[T_0] &= (I+H) D[HP]
         = (I+H) (D[H]\cap P{\mathcal H}),
\end{align}
as $T_0$ is the direct sum, $T_0=T_0\!\restriction_{P{\mathcal H}} \oplus\, 0\,$, in accordance with \eqref{HP}, \\
\text{\rm (ii)} and, in addition,
\begin{equation} \label{T0-eq}
  T_0\, g =g\,, \quad g \in D[T_0\!\restriction_{P{\mathcal H}}],
\end{equation}
thus $T_0$ can be extended to a bounded operator $\tilde{T}_0$ on ${\mathcal H}$ such that
\begin{align} \label{T0}
&\text{\it the closure of \,$T_0\!\restriction_{P{\mathcal H}}$\,
   is the identity operator\,$I_{P{\mathcal H}}$\,\,on\,} P{\mathcal H}, \nonumber\\
&\text{\it and}\,\,\tilde{T_0} = I_{P{\mathcal H}} \oplus 0 \,\,\,
 \text{\rm on }\,\, {\mathcal H}
              = P{\mathcal H} \oplus (P{\mathcal H})^{\perp}. \end{align}
In particular, $\tilde{T_0}$ is a contraction.
\end{lemma}
\begin{proof}
Product of linear operators $A$ and $B$ in ${\mathcal H}$ has the domain $D[AB] = B^{(-1)}D[A] := \{\,g \in D[B]:\, Bg \in D[A]\,\}$. Since $P$ commutes with $H_P$, we can rewrite \eqref{def-T0} as $T_0 := [(I+H)P][(I+ H_P)^{-1}P]$ with the domain
\begin{align*}
 D[T_0] &= [P(I+H_P)^{-1}]^{(-1)}D[(I+H)P] \\
        &= \{\,g \in D[(I+ H_P)^{-1}P]:\,(I+ H_P)^{-1}Pg \in D[(I+H)P]\} \\
        &= \{\,g \in D[H_P]:\,(I+ H_P)^{-1}Pg \in D[HP]\} \\
        &= P(I+H_P)D[HP]= P(I+H_P) (D[H]\cap P{\mathcal H}) \oplus (P{\mathcal H})^{\perp}
\end{align*}
and the range
\begin{align*}
 R[T_0] &= T_0D[T_0]   \\
        &= [(I+H)P][(I+ H_P)^{-1}P]P(I+H_P)D[HP] \\
        &= [(I+H)P]D[HP]
         = (I+H)(D[H]\cap P{\mathcal H}).
\end{align*}

Let us now turn to the claim (ii). Just as $H_P$ is the self-adjoint operator in ${\mathcal H}$ associated with the quadratic form $u \mapsto \|H^{1/2}Pu\|^2$ defined on $D[H^{1/2}P] = (D[H^{1/2}]\cap P{\mathcal H}) \oplus (P{\mathcal H})^{\perp}$, the self-adjoint operator $P(I+H_P)$ is associated with the form $u \mapsto \|(I+H)^{1/2}Pu\|^2$ defined in view of the inequalities
$$
 \frac{1}{\sqrt{2}}(I+H^{1/2}) \le (I+H)^{1/2} \le I+H^{1/2}.
$$
on the same domain $D[(I+H)^{1/2}P] = D[H^{1/2}P]$. Consequently,
\begin{equation} \label{I+H_P-1}
 P(I+H_P) = P(I+H)_P = ((I+H)^{1/2}P)^* (I+H)^{1/2}P\,.
\end{equation}
We are going to use it to show \eqref{T0-eq}. The adjoint $((I+H)^{1/2}P)^*$ to $(I+H)^{1/2}P$ is a closed extension of the closable (in general, non-closed) operator $P(I+H)^{1/2}$, in other words, $((I+H)^{1/2}P)^*\supset P(I+H)^{1/2}$, which yields
$$ 
  ((I+H)^{1/2}P)^*(I+H)^{1/2}P \supset
                  P(I+H)^{1/2}(I+H)^{1/2}P = P(I+H)P,
$$ 
i.e. the operator on the left-hand side is an extension of the operator on the right. As both sides are invertible, the analogous inclusion holds for their inverses,
$$ 
 [(I+H_P)^{-1}P]
= [((I+H)^{1/2}P)^*(I+H)^{1/2}P]^{-1} \supset [P(I+H)P]^{-1}.
$$ 
It follows that for any $g \in D[T_0]$ specified in \eqref{T0} we have
$$
   T_0\, g = (I+H)[(I+H_P)^{-1}P]g  = (I+H)[P(I+H)P]^{-1}g = Pg,
$$
which yields the desired claim \eqref{T0-eq}, since $T_0$ is reduced by the projection $P$ and  $D[T_0\!\restriction_{P{\mathcal H}}] \subseteq P{\mathcal H}$. Thus we see that the closure of $T_0\!\restriction_{P{\mathcal H}}$ is the identity operator $I_{P{\mathcal H}}$ on $P{\mathcal H}$, and $\tilde{T_0}$ as the closed extension of $T_0$ to the whole ${\mathcal H}$ has the norm not exceeding one.
\end{proof}

\begin{lemma} \label{l:estimate-T3}
The operator $T_3(s)$ in \eqref{T3}, acting in ${\mathcal H}$ with the domain
$P(I+isH_P)D[HP]
= (I+isH_P)(D[H]\cap P{\mathcal H})$, can be extended to a bounded operator
on the whole space ${\mathcal H}$ with the norm satisfying
$\|T_3(s)\|\leq \sqrt{2}$.
\end{lemma}
\begin{proof}
Note first that the claim of Lemma~\ref{l:T0} remains valid when we replace $H_P,\,H$ by $|s|H_P,\,|s|H$, respectively, for any $s \not= 0$.
We rewrite the operator $T_3(s)$ in \eqref{T3} as
\begin{align*}
T_3(s)
&= \big\{(I+ |s|H)[(I+ |s|H_P)]^{-1}P]\big\}      \\
 &\qquad\qquad \cdot \big\{[P(I+ |s|H_P)][(I+ isH_P)^{-1}P]\big\}\,.
\end{align*}
By spectral theorem the norms of the first and second factors on the right-hand side are one and $\sqrt{2}$, respectively, independently of $s$, and Lemma~\ref{l:T0} (ii) in combination with the above observation implies that the first factor is a contraction, which gives $\|T_3(s)\| \leq \sqrt{2}$.
\end{proof}

\smallskip

In the next step, we turn to investigation of the second and third factors, \eqref{T2tau} and \eqref{T3tau}, of the operator $T^P(t,s;\tau)$ in \eqref{split-D(t,s;tau)to123}. We begin with proving a crucial property of the $\tau$-family $\{T_3(s;\tau)\}_{0<\tau\leq 1}$ with $s \in {\mathbb R} \setminus (M \cup \{0\})$, where the set $M \subseteq {\mathbb R} $ of Lebesgue measure zero was introduced in the text following \eqref{a.e.Cov}; recall that we adopted the separability assumption. After doing that, we will focus on the $\tau$-family $\{T_2(t,s;\tau)\}_{0<\tau\leq 1}$ with $t,\,s \in {\mathbb R} \setminus \{0\}$.

\begin{lemma}  \label{l:unif-bound-T3}
Let the Hilbert space ${\mathcal H}$ be separable and consider a number
$s \in {\mathbb R} \setminus (M \cup \{0\})$, so that for every vector
$f \in {\mathcal H}$ the $\tau$-family
$\{[P(\tau)(I+sK(s\tau))P(\tau)]^{-1}f\}$ converges in the Hilbert
space norm to $[(I+isH_P)^{-1}P]f$ as $\tau\to 0+$.
Then the following claims are valid:

\text{\rm (i)} For fixed $s$, the operator family $\{T_3(s;\tau)\}_{0<\tau\leq 1}$ defined by \eqref{T3tau} is uniformly bounded on ${\mathcal H}$, and converges strongly to $T_3(s)$ in \eqref{T3} as $\tau\to 0+$.

\text{\rm (ii)} To be specific, for a fixed $s \in {\mathbb R} \setminus (M \cup \{0\})$ there is a constant $C_{T_3}(s)\ge\sqrt{2}$ such that
\begin{equation}
\|T_3(s;\tau)\| \leq C_{T_3}(s) \quad \text{for all}\;\: \tau \in (0,1]\,.
                                                     \label{unif-bound-T3}
\end{equation}
\end{lemma}

\medskip

\begin{proof}
(i) Our aim is to verify that $T_3(s;\tau)g$ converges for $g \in {\mathcal H}$ to $T_3(s)g$ as $\tau\to 0+$ in the Hilbert space norm and to use this fact to establish the uniform boundedness of $\{T_3(s;\tau)g\}_{0<\tau\leq 1}$.

To this end, let us recall the notion of convergence in the strong graph sense: given a sequence $\{A_n\}_{n=1}^{\infty}$ of operators in a Hilbert space ${\mathcal H}$ we say that $(\psi,\varphi) \in {\mathcal H}\times {\mathcal H}$ belongs to the strong graph limit if one can find $\psi_n \in D[A_n],\, n= 1,2, \dots,$ such that
$$
  \psi_n \longrightarrow \psi,\quad
  A_n \psi_n \longrightarrow \varphi
  \quad\;\text{\rm as}\;\; n\to\infty\,.
$$
We denote the set of such pairs $(\psi,\varphi)$ by $\Gamma_{\infty}^s$. If $\Gamma_{\infty}^s$ is the graph of an operator $A$, i.e. $\Gamma_{\infty}^s = \{ (\psi, A \psi) \in {\mathcal H}\times {\mathcal H};\, \psi \in D[A] \}$, $\:A$ is also said to be the \emph{strong graph limit} of $\{A_n\}$. We have the following result \cite[Thm~4]{GJ69}, see also \cite[Thm~VIII.26]{RS80}:

\begin{proposition} \label{p:graphlim}
Let $\{A_n\}_{n=1}^{\infty}$ and $A$ be self-adjoint operators in ${\mathcal H}$. Then $\{A_n\}$ converges to $A$ in strong resolvent sense if and only if $A$ is the strong graph limit of $\{A_n\}$.
\end{proposition}

To make use of this result, let us consider the family $\{A_{\tau}\}_{0< \tau \leq 1}$ of the following bounded and self-adjoint operators,
$$
A_{\tau} : = I+ |s|H_{\tau}\,,
                                            \quad\; 0< \tau \leq 1.
$$
By spectral theorem its elements converge in the strong operator topology to the self-adjoint operator $I+|s|H$ as $\tau\to 0+$, and that by Proposition~\ref{p:graphlim} means that $I+|s|H$ is the strong graph limit of $\{A_{\tau}\}$. Since $s \in {\mathbb R} \setminus (M\cup \{0\})$ holds by assumption, we have
\begin{align*}
 &\quad \psi_{\tau} := [P(\tau)(I+ sK(s\tau))P(\tau)]^{-1}g
 \longrightarrow  \psi := [(I+isH_P)^{-1}P]g\,,  \\
\intertext{as $\tau\to 0+$ for $g\in {\mathcal H}$, with}
 &\quad [(I+isH_P)^{-1}P]g  \in D[H_P]
       \subset D[(I+ |s|H)P] \subset D[I+ |s|H]\,.
\end{align*}
In this way, the pair $(\psi, (I+ |s|H)\psi)$ belongs to the set $\Gamma_{\infty}^s$ resulting from the strong graph limit of the family $\{A_{\tau}\}$, and consequently, $A_{\tau}\psi_{\tau} \longrightarrow (I+ |s|H)\psi$, i.e.
\begin{align*}
&A_{\tau}[P(\tau)(I+ sK(s\tau))P(\tau)]^{-1}g
 \longrightarrow (I+ |s|H)[(I+isH_P)^{-1}P]g\,,
\end{align*}
in other words, the sought convergence, $T_3(s;\tau)g \longrightarrow\, T_3(s)g$. This implies, by the uniform boundedness principle (or Banach-Steinhaus theorem), the family $\{T_3(s;\tau)\}_{0\le\tau< 1}$ is uniformly bounded in the operator norm.

(ii) Consider again an arbitrary $g \in {\mathcal H}$. In view of the convergence established above in combination with the bound on the norm of $T_3(s)$ from Lemma~\ref{l:estimate-T3}(ii), we see that for a fixed $s \in {\mathbb R}\setminus (M\cup\{0\})$ there is a positive, $g$-dependent number $\tau_s = \tau_s(g) \leq 1$ such that
\begin{align*}
\|T_3(s;\tau)g\| \leq \|T_3(s)g\| + \tfrac12
                 \leq \|T_3(s)\|\|g\| + \tfrac12
                 \leq \sqrt{2}\|g\| + \tfrac12\,,
                                             \;\: \tau \in (0,\tau_s]\,.
\end{align*}
This helps us to complement our knowledge of the uniform boundedness using
the estimate
$$
\sup_{0<\tau\leq 1} \|T_3(s;\tau)g\|
\leq \max\big\{\sqrt{2}\|g\| + \tfrac12,\,
                   \sup_{\tau_s<\tau\leq 1}\|T_3(s;\tau)g\|\,\big\}\,.
$$
We note that the function $\tau\mapsto\|T_3(s;\tau)g\|$ is continuous on the compact set $[\tau_s,1] \subset {\mathbb R}$, hence the right-hand side is bounded for any $g$ and, applying again the uniform boundedness principle, we conclude that the $\tau$-family $\{T_3(s;\tau)\}_{0<\tau\leq 1}$ is uniformly bounded by a constant $C_{T_3}(s) \geq \sqrt{2}$ depending on $s$. This completes the proof of Lemma~\ref{l:unif-bound-T3}.
\end{proof}

\medskip

It remains to deal with the middle factor $T_2(t,s;\tau)$ in \eqref{T2tau}. We rewrite it as
\begin{align}
T_2(t,s;\tau)
&= i(sH(s\tau) -tH(t\tau))(I+ |s|H_{\tau})^{-1}  \nonumber\\
 &\quad +(sG(s\tau) -tG(t\tau))I+ |s|H_{\tau})^{-1} \nonumber\\
&=: iT_{2,H}(t,s;\tau) + T_{2,G}(t,s;\tau).     \label{T2HGtau}
\end{align}
For a fixed $\tau\in(0,1]$, the operators $T_{2,H}(t,s;\tau)$, $T_{2,G}(t,s;\tau)$, as well as their linear combination, are bounded operators which are in the functional calculus sense obtained from $H$ as its  appropriate bounded and continuous functions; in view of \eqref{tH(t)-sH(s)} and \eqref{tG(t)-sG(s)} we have
\begin{align*}
T_{2,H}(t,s;\tau)
&= \phi_H(t,s;\tau,H)
 := (tH(t\tau)-sH(s\tau))(I+ |s|H_{\tau})^{-1}         \\
&= \frac{\tfrac{2}{\tau} \sin(\tfrac{t-s}2\tau\lambda)
   \cos(\tfrac{s+t}2\tau\lambda)}{I+ |s|\lambda(I+\tau\lambda)^{-1}}\,, 
    \\
T_{2,G}(t,s;\tau)
&=  \phi_G(t,s;\tau,H)
  := (tG(t\tau)- sG(s\tau))(I+ |s|H_{\tau})^{-1}        \nonumber\\
&= \frac{\tfrac{2}{\tau} \sin(\tfrac{t-s}2\tau\lambda)\,
   \sin(\tfrac{s+t}2\tau\lambda)}{I+ |s|\lambda(I+\tau\lambda)^{-1}}\,,  
\end{align*}
and
$$
T_2(t,s;\tau)
=  \phi(t,s;\tau,H)
 := i\phi_H(t,s;\tau,H) +\phi_G(t,s;\tau,H)       
$$
It is straightforward to check that
\begin{equation}
|\phi(t,s;\tau,\lambda)|^2 = \Big|\frac{\tfrac{2}{\tau} \sin(\tfrac{t-s}2\tau\lambda)}
     {I+ |s|\lambda(I+\tau\lambda)^{-1}}\Big|^2,    \label{bound-phi}
\end{equation}
which means that for any $g \in {\mathcal H}$ we have
\begin{align}
\|T_2(t,s;\tau)g\|^2 &= \|T_{2,H}(t,s;\tau)g\|^2 + \|T_{2,G}(t,s;\tau)g\|^2 \nonumber\\
&= \int_{0-}^{\infty} \Big|\frac{\tfrac{2}{\tau} \sin(\tfrac{t-s}2\tau\lambda)}
          {I+ |s|\lambda(I+\tau\lambda)^{-1}}\Big|^2
                           \|E(\mathrm{d}\lambda)g\|^2\,. \label{norm-T2}
\end{align}
The following lemma represents the second crucial element in establishing Lemma~\ref{l:equiconti}, and in this way, proving our main theorems.

\begin{lemma} \label{l:estimate-T}
Consider $t,s \in {\mathbb R} \setminus\{0\}$ such that either $t,s >0$ or $t,s < 0$. Then for $T_2(t,s;\tau)$ given by \eqref{T2tau} and for $D(t,s;\tau)$ given by \eqref{D(t,s;tau)} or \eqref{split-D(t,s;tau)to123} the following holds:

\smallskip

\noindent \text{\rm (i)} The $\tau$-family $\{T_2(t,s;\tau\}$ of bounded operators in \eqref{T2HGtau} is uniformly bounded locally uniformly for $t,s \in {\mathbb R} \setminus\{0\}$. Further, for every $f \in {\mathcal H}$ and any $\varepsilon>0$, there is an $s$-dependent number  $\delta = \delta(f;\varepsilon;s) >0$ such that
\begin{equation}
t\in {\mathbb R}\setminus \{0\}\;\; \text{\it with}\;\, |t-s|< \delta \;\;
\Longrightarrow\;\; \|T_{2}(t,s;\tau)f\| < \varepsilon,  \label{equicont-T2}
\end{equation}
uniformly with respect to $\tau\in(0,1]$.

\smallskip

\noindent \text{\rm (ii)} Let further $s \in {\mathbb R}\setminus (M\cup\{0\})$, then for every $f \in {\mathcal H}$ and any $\varepsilon>0$,
there is a number \mbox{$\delta = \delta(f;\varepsilon;s) >0$} such that
\begin{equation}
t\in {\mathbb R}\setminus \{0\}\;\;
                           \text{\it with}\;\; |t-s|< \delta
\;\;\Longrightarrow\;\;
 \|D(t,s;\tau)f\|< \,  \varepsilon,      \label{equicont-D}
\end{equation}
uniformly with respect to $\tau\in(0,1]$.
\end{lemma}
\begin{proof}
(i) The $\tau$-families of bounded operators appearing in \eqref{T2HGtau} converge in the strong operator topology for any fixed pair $t,s \in {\mathbb R} \setminus\{0\}$; by functional calculus we find easily that
\begin{align*}
&T_{2,H}(t,s;\tau) \overset s \to T_{2,H}(t,s)  :=\tfrac{(t-s)H}{I+|s|H},
 \quad T_{2,G}(t,s;\tau) \overset s \to T_{2,G}(t,s) := 0, \\
&T_2(t,s;\tau) \overset s \longrightarrow
 T_2(t,s) := iT_{2,H}(t,s) + T_{2,G}(t,s)  
\end{align*}
holds as $\tau\to 0+$, and therefore
$$ 
 |T_2(t,s;\tau)|^2 \overset s \longrightarrow
 |T_2(t,s)|^2 = T_{2,H}(t,s)^2 + T_{2,G}(t,s)^2
               = \big(\tfrac{(t-s)H}{I+|s|H}\big)^2. 
$$ 
Our way to prove the claim is to apply the uniform boundedness principle to these operator families. From \eqref{norm-T2} we obtain
$$ 
 \|T_2(t,s;\tau)\| = \sup_{\|g\|=1}\|T_2(t,s;\tau)g\| = \sup_{\lambda\geq 0}|\phi(t,s;\tau,\lambda)|\,, 
$$ 
and for $T_2(t,s) = \phi(t,s;H)$ with $\phi(t,s;\lambda):= \tfrac{(t-s)\lambda}{1+|s|\lambda}$ we have similarly
\begin{equation}
 \|T_2(t,s)\| = \sup_{\lambda\geq 0}|\phi(t,s;\lambda)| = \tfrac{|t-s|}{|s|}\,. \label{T2ts-norm}
\end{equation}
We note that $\phi(t,s;\tau,\lambda)$ is continuous with respect to all the variables, $\tau \in (0,1]$, $\lambda \in [0,\infty)$, and $t,s \in {\mathbb R} \setminus\{0\}$. Our intention is now to regard the family $\{T_2(t,s;\tau)\}$ as depending on a broader set of parameters: we fix a nonzero $s$ and consider $\tau$ in $(0,1]$ and $t$ in a neighborhood of $s$, more specifically,
\begin{equation}  \label{larger-parameter-set}
 \tau \in (0,1], \quad
  t \in D(s) := \{t \in {\mathbb R}\setminus \{0\};\, |t-s|
             \leq \tfrac{|s|}2\}.
\end{equation}
For each fixed $f \in {\mathcal H}$ and $s \in {\mathbb R} \setminus\{0\}$ we put
\begin{equation}
\Phi_{f,s}(t;\tau) := \left\{
  \begin{array}{lcc}
    \|T_2(t,s;\tau)f\| & \quad \dots\;\; &  0<\tau\leq 1, \\[.3em]
    \|T_2(t,s)f\|      & \quad \dots\;\; & \tau=0,
\end{array}\right.
\end{equation}
which is a bounded and continuous function on the compact set $D(s) \times [0,1] \subset {\mathbb R}^2$. Then
$
C_{f,s} := \max_{t,\tau \in D(s) \times [0,1]}\Phi_{f,s}(t;\tau)
$
exists, depending on $f$ and $s$, but being independent of $t \in D(s)$ and $\tau \in(0,1]$, and from the strong convergence of $\{T_2(t,s;\tau)\}$ to $T_2(t,s)$ noted above we infer that
$$
   \|T_2(t,s;\tau)f\| \leq C_{f,s}
$$
holds for given $f \in {\mathcal H}$,  $\:s\ne 0$, and all $(t,\tau) \in D(s) \times (0,1]$. Applying now the uniform boundedness principle to this operator family with the extended set of parameters \eqref{larger-parameter-set}, we establish the existence of a constant $C_{T_2}(s)>0$, depending on $s$ but independent of $t \in D(s)$ and $\tau \in (0,1]$, such that
\begin{equation}  \label{T2-bound}
   \|T_2(t,s;\tau)\| \leq C_{T_2}(s).
\end{equation}
In view of the strong convergence, it cannot be smaller than the norm of $T_2(t,s)$ in \eqref{T2ts-norm}, that is, $C_{T_2}(s) \geq \tfrac{|t-s|}{|s|}$.

Since $E([\lambda,\infty)) \overset s \to 0$ holds as $\lambda\to\infty$, one can find for a fixed $s$, any $f \in {\mathcal H}$, and each $\varepsilon>0$ a number $\lambda_0 = \lambda_0(f;\varepsilon;s) >0$ such that
\begin{equation} \label{phi-tail}
  \|E([\lambda_0,\infty))f\|^2
                < \tfrac1{C_{T_2}(s)^2}\cdot\tfrac{\varepsilon^2}{2}\,.
\end{equation}
In this case using \eqref{bound-phi} and \eqref{phi-tail}, we get
\begin{align*}
\|T_2&(t,s;\tau)f\|^2
= \|\phi(t,s;\tau,H)f\|^2                             \\
&= \Big(\int_{0-}^{\lambda_0}\,\, +\,\, \int_{\lambda_0}^{\infty}\Big)
   |\phi(t,s;\tau,\lambda)|^2 \|E(\mathrm{d}\lambda)f\|^2 \\
&\leq \int_{0-}^{\lambda_0}
  \bigg|\tfrac{\tfrac{2}{\tau}\sin(\tfrac{t-s}2\tau\lambda)}
 {1+ (s\lambda)(I+\tau\lambda)^{-1}}\bigg|^2
  \|E(\mathrm{d}\lambda)f\|^2
               +C_{T_2}(s)^2\|E([\lambda_0,\infty))f\|^2  \\
&\leq  (|t-s|\lambda_0)^2 \|E([0,\lambda_0))f\|^2
   + \tfrac{\varepsilon^2}2\,,
\end{align*}
where the last inequality comes from the estimate of the numerator of the squared modulus factor in the above integral,
$$
 \big|\tfrac{2}{\tau}\sin(\tfrac{t-s}2\tau\lambda)\big|
\leq |t-s|\lambda_0, \qquad 0\leq \lambda \leq \lambda_0.
$$
Since the set $\{(\tau,\lambda); \, 0<\tau\leq 1,\, 0\leq \lambda \leq \lambda_0\}$ is bounded, and therefore its closure is compact, we infer that for a fixed $s\in {\mathbb R}\setminus \{0\}$ there is a number $\delta = \delta(f;\varepsilon;s) >0$ such that
$$
  t \in {\mathbb R}\setminus \{0\}\,\, \text{\rm with}\,\, |t-s|<\delta \,\,
  \Longrightarrow
  (|t-s|\lambda_0)^2 (1+\|f\|^2) < \tfrac{\varepsilon^2}{2}
$$
uniformly for $0<\tau\leq 1$. In fact, we can find it explicitly; it is enough to choose
$$
\delta(f;\varepsilon;s)
:= \min\big\{\tfrac1{\sqrt{2}\lambda_0(1+\|f\|^2)^{1/2}}\,\varepsilon,
                                                     \tfrac{|s|}2\big\}.
$$
Thus we have
$$
\|T_2(t,s;\tau)f\|^2
\leq (1+\|f\|^2)^{-1}
        \tfrac{\varepsilon^2}{2} \|E([0,\lambda_0))f\|^2
                                   + \tfrac{\varepsilon^2}2
  < \varepsilon^2,
$$
in other words,
$$
 \|T_2(t,s;\tau,H)f\| < \varepsilon,
$$
which yields the implication \eqref{equicont-T2}.

\smallskip\noindent

(ii) To prove the remaining claim of the lemma, we first note that $T_1(t;\tau) \equiv T^P(t;\tau)$ is a contraction, $\|T_1(t;\tau)\| \leq 1$, and therefore
\begin{align*}
 \|D(t,s;\tau)f\|
&= \|T^P(t,s;\tau)f\| = \|T_1(t;\tau) T_2(t,s;\tau) T_3(s;\tau)f\|\\
&\leq \|T_2(t,s;\tau) T_3(s;\tau)f\|.
\end{align*}
To deal with the last norm we repeat the argument from the part (i) replacing $f$ by $T_3(s;\tau)f$, obtaining for any $\lambda_0>0$ the estimate
\begin{align}
&\|T_2(t,s;\tau) T_3(s;\tau)f\|^2                \nonumber   \\
&= \Big(\int_{0-}^{\lambda_0} +\int_{\lambda_0}^{\infty} \Big)
|\phi(t,s;\tau,\lambda)|^2\|E(\mathrm{d}\lambda) T_3(s;\tau)f\|^2
                                                   \nonumber\\
&\leq \Big[\int_{0-}^{\lambda_0}
 \Big|\tfrac{\tfrac{2}{\tau} \sin\tfrac{(t-s)}2\tau\lambda}
     {1+ (s\lambda)(I+\tau\lambda)^{-1}}\Big|^2
\|E(\mathrm{d}\lambda) T_3(s;\tau)f\|^2 \nonumber \\
&\qquad\qquad\qquad\qquad\qquad
+ C_{T_2}(s)^2\|E([\lambda_0,\infty))T_3(s;\tau)f\|^2\Big] \label{T_2T_3} \\
&\leq (|t-s|\lambda_0)^2 \|E([0,\lambda_0))T_3(s;\tau)f\|^2
 + C_{T_2}(s)^2\|E([\lambda_0,\infty)) T_3(s;\tau)f\|^2 \nonumber
\end{align}
with the constant $C_{T_2}(s)$ from \eqref{T2-bound}. The second term of the last expression converges to zero as $\lambda_0 \to \infty$ for any \emph{fixed} $\tau\in(0,1]$.

To make use of the estimate, however, we need to find a $\lambda_0$ independent of $\tau$ in order to obtain convergence \emph{uniform} with respect to $0 <\tau \leq 1$. This can be achieved, however, because for fixed
$f \in {\mathcal H}$ and $s \in {\mathbb R} \setminus (M\cup \{0\})$, the set
$$
 V_f := \{ T_3(s;\tau)f ;\,\, 0 <\tau \leq 1 \} \subset {\mathcal H}
$$
is bounded by \eqref{unif-bound-T3}, and has a compact closure, as a consequence of the strong continuity of $T_3(s;\tau)f$ with respect to the variable $\tau>0$. Then for every $f \in {\mathcal H}$ and every $\varepsilon>0$ one can find an $s$-dependent number $\lambda_{00} = \lambda_{00}(f;\varepsilon;s) >0$ such that
\begin{equation} \label{uniform-for-tau}
 \sup_{0<\tau \leq 1}\|E([\lambda_{00},\infty))T_3(s;\tau)f\|^2
    <  \tfrac1{C_{T_2}(s)^2}\cdot\tfrac{\varepsilon^2}{2}
\end{equation}
with the constant $C_{T_2}(s)$ in \eqref{T2-bound}. Let us be more explicit about the last claim. Since $V_f$ is totally bounded, there is a finite family of vectors $\{y_j\}_{j=1}^N \subset V_f$, $y_j = T_3(s;\tau_j)f$ for some
$\tau_j\in(0,1]$, and an open ball $B(0; \tfrac{\varepsilon}{\sqrt{8}C_{T_2}(s)})$ with the center at the origin and radius $\tfrac{\varepsilon}{\sqrt{8}C_{T_2}(s)}$ for which we have
$
 V_f\,\, \subset\,\, \cup_{j=1}^N \big(y_j
    + B(0; \tfrac{\varepsilon}{\sqrt{8}C_{T_2}(s)})\big).
$
Using again the fact that $E([\lambda,\infty)) \overset s \to 0$ holds as $\lambda\to\infty$ we infer that there is a family $\{\lambda_j\}_{j=1}^N$ of large positive numbers such that $C_{T_2}(s)\|E([\lambda_j,\infty)) y_j\| \leq \tfrac{\varepsilon}{\sqrt{8}}$ holds for $j=1,2,\dots, N$ and we put $\lambda_{00} := \max_{j=1,2,\dots, N}{\lambda_j}$. The finite union coverage means that any $\phi \in V_f$ satisfies $\phi \in y_{j_{\phi}} + B(0; \tfrac{\varepsilon}{\sqrt{8}C_{T_2}(s)})$ for some $1 \leq j_{\phi} \leq N$ with $0< \tau_{j_{\phi}} \leq 1$. Noting that $\|E([\lambda_{00},\infty))\|\le 1$, we have
\begin{align*}
&C_{T_2}(s)^2\|E([\lambda_{00},\infty))\phi\|^2      \\
&\leq C_{T_2}(s)^2\big(\max_{j=1,2,\dots, N}\|E([\lambda_{00},\infty))y_j\|
              +  \tfrac{\varepsilon}{\sqrt{8}C_{T_2}(s)}\big)^2 \\
&\leq 2C_{T_2}(s)^2\big[\max_{j=1,2,\dots, N}\|E([\lambda_{00},\infty))y_j\|^2
              + \big( \tfrac{\varepsilon}{\sqrt{8}C_{T_2}(s)})^2\big] \\
&< \tfrac{\varepsilon^2}4 + \tfrac{\varepsilon^2}4
= \tfrac{\varepsilon^2}2.
\end{align*}

\smallskip

\noindent Returning to the first term on the right-hand side of \eqref{T_2T_3} where take $\lambda_0=\lambda_{00}$, and noting that $\|T_3(s;\tau)f\| \leq C_{T_3}(s)\|f\|$ holds in view of \eqref{unif-bound-T3}, we see that there exists a number $\delta = \delta(f;\varepsilon;s) >0$ such that for a fixed $s$ we have
$$
  t \in {\mathbb R}\setminus \{0\}\; \text{\rm with}\;
|t-s|<\delta \;
  \Longrightarrow
  (|t-s|\lambda_{00})^2 (C_{T_3}(s)\|f\|)^2 < \tfrac{\varepsilon^2}{2},
$$
uniformly for all $\tau\in(0,1]$; one can be again explicit and choose
$$
\delta(f;\varepsilon;s) :=
\min\{\tfrac{1}{\sqrt{2}\lambda_{00}(1+(C_{T_3}(s)\|f\|)^2)^{1/2}}\,
      \varepsilon,\, \tfrac{|s|}2\}\,.
$$
The estimate \eqref{T_2T_3} then gives
$$
\|T_2(t,s;\tau) T_3(s;\tau)f\|^2
< \tfrac{\varepsilon^2}2 + \tfrac{\varepsilon^2}2
  = \varepsilon^2,
$$
or in other words
$$
 \|D(t,s;\tau)f\| \leq \|T_2(t,s;\tau) T_3(s;\tau)f\|
< \varepsilon,
$$
which yields \eqref{equicont-D}. This concludes the proof of the present lemma,
and in this way also of Lemma~\ref{l:equiconti} as we have outlined in the opening of this section.
\end{proof}

\section{An example}
\label{s: example}

Let us finally mention briefly a typical situation in which Zeno
dynamics occurs as indicated, for instance, in the
paper~\cite{FPSS01}, namely the perpetual position ascertaining. We
consider an open domain $\Omega\subset\mathbb{R}^d $ with a smooth
boundary, thought of as the detector volume, and associate with it
the orthogonal projection $P$ on $L^2(\mathbb{R}^d)$ acting as the
multiplication operator by the indicator function $\chi_\Omega$ of
the set $\Omega$. Suppose that the dynamics of the particle
undisturbed by the measurement is free, that is, described by the
Hamiltonian $H = -\Delta$, i.e. the Laplacian in $\mathbb{R}^d$
which is a nonnegative self-adjoint operator in $L^2(\mathbb{R}^d)$.
The assumption of density of the domain of $H^{1/2}P =
(-\Delta)^{1/2}\chi_{\Omega}$ is satisfied, since it contains
$C_0^{\infty}(\Omega) \cup C_0^{\infty}({\mathbb R}^d \setminus
\overline{\Omega})$, where $\overline{\Omega}$ is the closure of
$\Omega$, and this family of functions is dense in $L^2({\mathbb
R}^d)$.

Consider further the Dirichlet Laplacian $-\Delta_\Omega$ in
$L^2(\Omega)$ defined in the usual way \cite[Thm~XIII.15]{RS78} as the
Friedrichs extension of the appropriate quadratic form.
It can be checked \cite[Prop.~6.1]{EI05} that
$$
(-\Delta)_P := ((-\Delta)^{1/2}P)^* (-\Delta)^{1/2}P
$$
is densely defined and its restriction to the subspace $L^2(\Omega)$ is nothing
but the Dirichlet Laplacian $-\Delta_\Omega$ with the domain
$D[-\Delta_\Omega] = W_0^1(\Omega) \cap W^2(\Omega)$, and
$$
  -\Delta_\Omega = (-\Delta)_P\restriction_{PL^2({\mathbb R}^d)}
                 = (-\Delta)_P\restriction_{L^2(\Omega)}\,.
$$
Then Theorem~\ref{th:main} says that
$$
 \text{s\,-}\!\!\lim_{n\rightarrow \infty}  (P\,\mathrm{e}^{-it(-\Delta/n)}P)^n
     = \mathrm{e}^{-it(-\Delta_\Omega)}P
$$
holds in the strong operator topology of $\mathcal{B}(L^2(\mathbb{R}^d))$, the Banach space of the bounded linear operators on $L^2(\mathbb{R}^d)$, in other words, that the perpetual reduction of the wave function forces the particle to move within the region $\Omega$ as if its boundary was Dirichlet, i.e. hard wall. This is, of course, the expected conclusion indicated, e.g. in \cite{FP08, FPSS01}, however, only Theorem~\ref{th:main} allows one to state such a claim with the proper rigor.


\subsection*{Acknowledgments}

The research was supported by the Czech Science Foundation within the project 17-01706S and in part by Grant-in Aid for Scientific Research 16K05230, Japan Society for the Promotion of Science, and by the EU project CZ.02.1.01/0.0/0.0/16\textunderscore 019/0000778. The authors are grateful to Tsuyoshi Ando for valuable discussions, to Hiroshi Tamura, Valentin Zagrebnov, and late Hagen Neidhardt for a number of useful comments, and to Hideo Tamura for his unceasing encouragement.



\begin{thebibliography}{99}

\bibitem{BN67}
A.~Beskow, J.Nilsson: The concept of wave function and the irreducible representations of the Poincar\'{e} group, II.~Unstable systems and the exponential decay law, \emph{Arkiv Fys.} \textbf{34} (1967), 561--569.


\bibitem{Ch68}
P.R.~Chernoff: Note on product formulas for operator semigroups, \emph{J. Funct. Anal.} \textbf{2} (1968), 238--242.

\bibitem{Ch74}
P.R.~Chernoff: \emph{Product Formulas, Nonlinear Semigroups, and Addition of Unbounded Operators}, Memoirs of the American Mathematical Society, vol.~140; AMS, Providence, R.I. 1974.

\bibitem{Ex85}
P.~Exner: \emph{Open Quantum Systems and Feynman Integrals}, Fundamental Theories of Physics, vol.~6; D.~Reidel Publ. Co., Dordrecht 1985.

\bibitem{EI05}
P.~Exner and T.~Ichinose: A product formula related to quantum Zeno dynamics, \emph{Ann. Henri Poincar\'e} \textbf{6}(2) (2005), 195--215.

\bibitem{EINZ07}
P.~Exner, T.~Ichinose, H.~Neidhardt, and V.A.~Zagrebnov: Zeno product formula revisited, \emph{Integral Eq. Oper. Theory} \textbf{57}(1) (2007), 67--81.


\bibitem{FP08}
P.~Facchi and S.~Pascazio: Quantum Zeno dynamics: mathematical and physical aspects, \emph{J. Phys. A: Math. Theor.} \textbf{41}(49) (2008), 493001.

\bibitem{FPSS01}
P.~Facchi, S.~Pascazio, A.~Scardicchio, and L.S.~Schulman: Zeno
dynamics yields ordinary constraints, \emph{Phys. Rev.} \textbf{A65}
(2001), 012108.

\bibitem{Fr72}
C.N.~Friedman: Semigroup product formulas, compressions and continuous observation in quantum mechanics, \emph{Indiana Univ. Math. J.} \textbf{21} (1972), 1001--1011.


\bibitem{GJ69}
J.~Glimm, A.~Jaffe: Singular perturbations of selfadjoint operators, \emph{Comm. Pure Appl. Math.} \textbf{22} (1969), 401--414.

\bibitem{Ha67}
P.~Halmos: \emph{A Hilbert space problem book},
D. Van Nostrand Co., Inc., Princeton, N.J. -Toronto, Ont. -London 1967.


\bibitem{Ho04}
A.~Hodges: What would Alan Turing have done after 1954?, in
\emph{Alan Turing: Life and Legacy of a Great Thinker}
(Ch.~Teuscher, ed.), Springer, Heidelberg 2004; pp.~43--58.





\bibitem{Ich15}
T.~Ichinose: The modified unitary Trotter--Kato and Zeno product formulas revisited, in {\em Functional Analysis and Operator Theory for Quantum Physics} (J.Dittrich, H.Kova\v{r}\'{\i}k, A.Laptev, eds.), EMS Publ., Z\"urich 2017; p.~401--417.



\bibitem{IHBW90}
W.M.~Itano, D.J.~Heinzen, J.J.~Bollinger, and D.J.~Wineland:
Quantum Zeno effect, \emph{Phys. Rev.} \textbf{A41} (1990), 2295--2300.


\bibitem{Ka76}
T.~Kato: \emph{Perturbation Theory for Linear Operators}, 2nd edition, Springer, Berlin 1976.


\bibitem{Ka78}
T.~Kato: Trotter's product formula for an arbitrary pair of self-adjoint contraction semigroups, in \emph{Topics in functional analysis (essays dedicated to M. G. Krein on the occasion of his 70th birthday)}, Adv. Math. Suppl. Stud., vol.~3, pp. 185--195; Academic Press, New York 1978

\bibitem{KeNa76}
J.L.~Kelley and I.~Namioka: \emph{Linear topological spaces}, with
the collaboration of W.F.~Donoghue, Jr., Kenneth R. Lucas,
B.J.~Pettis, Ebbe Thue Poulsen, G.~Baley Price, Wendy Robertson,
W.R.~Scott, and Kennan~T. Smith. Second corrected printing. Graduate
Texts in Mathematics, vol.~36. Springer, New York-Heidelberg 1976

\bibitem{Ko69}
G.~K\"othe: \emph{Topologische lineare R\"aume. I},
Die Grundlehren der mathematischen Wissenschaften, Band 107,  Springer, Berlin-Heidelberg-New York 1966; \emph{Topological vector spaces. I}, Band 159, English translation by D.J.H.~Garling, 1969

\bibitem{MS03}
M.~Matolcsi and R.~Shvidkoy: Trotter's product formula for projections, \emph{Arch. der Math.} \textbf{81} (2003), 309--317.

\bibitem{MS77}
B.~Misra and E.C.G.~Sudarshan: The Zeno's paradox in quantum theory,
\emph{J. Math. Phys.} \textbf{18} (1977), 756--763.



\bibitem{RS80}
M.~Reed and B.~Simon: \emph{Methods of Modern Mathematical Physics I:
Functional Analysis, Revised and Enlarged Edition}, Academic Press,
New York, London 1980.


\bibitem{RS78}
M.~Reed and B.~Simon: \emph{Methods of Modern Mathematical Physics IV:
Analysis of Operators}, Academic Press, New York, London 1978

\bibitem{Si15}
B.~Simon: \emph{Operator Theory. A Comprehensive Course of Analysis, Part 4}, AMS, Providence, R.I. 2015


\bibitem{Tr59}
H.F.~Trotter: On the product of semi-groups of operators, \emph{Proc. Am. Math. Soc.} \textbf{10} (1959), 545--551.

\end{thebibliography}
\end{document}